\documentclass[twoside,11pt]{article}

\usepackage{caption}
\usepackage{jmlr2e}
\usepackage{xcolor}
\usepackage{amsmath}
\usepackage{mathtools}
\makeatletter
\let\proof\@undefined
\let\endproof\@undefined
\@for\sgubu@env:=theorem,proposition,lemma,corollary,definition,assumption,remark\do{%
  \expandafter\let\csname\sgubu@env\endcsname\@undefined
  \expandafter\let\csname end\sgubu@env\endcsname\@undefined
  \expandafter\let\csname c@\sgubu@env\endcsname\@undefined
  \expandafter\let\csname the\sgubu@env\endcsname\@undefined
  \expandafter\let\csname p@\sgubu@env\endcsname\@undefined
  \expandafter\let\csname cl@\sgubu@env\endcsname\@undefined
  \expandafter\let\csname theH\sgubu@env\endcsname\@undefined
}
\makeatother
\usepackage{amsthm}
\usepackage{thmtools,thm-restate}
\makeatletter
\@ifundefined{newcounteralias}{}{%
  \renewcommand\thmt@autorefsetup{%
    \@xa\def\csname\thmt@envname autorefname\@xa\endcsname\@xa{\thmt@thmname}}%
}
\makeatother
\usepackage{subcaption}
\usepackage[shortlabels]{enumitem}
\usepackage{booktabs}
\usepackage{appendix}
\usepackage{microtype}
\theoremstyle{plain}
\newtheorem{theorem}{Theorem}[section]
\newtheorem{proposition}[theorem]{Proposition}
\newtheorem{lemma}[theorem]{Lemma}
\newtheorem{corollary}[theorem]{Corollary}

\theoremstyle{definition}
\newtheorem{definition}[theorem]{Definition}
\newtheorem{assumption}[theorem]{Assumption}
\newtheorem{remark}[theorem]{Remark}
\newtheorem{fact}{Fact}
\newcommand{\R}{\mathbb{R}}
\newcommand{\N}{\mathbb{N}}
\newcommand{\E}{\mathbb{E}}

\DeclareMathOperator{\Var}{Var}
\DeclareMathOperator{\Cov}{Cov}

\newcommand{\op}{\mathrm{op}}
\DeclareMathOperator{\tr}{tr}

\newcommand{\hess}{\nabla^2}
\DeclarePairedDelimiter\tond{(}{)}
\providecommand{\Tilde}{\widetilde}
\newcommand{\pete}[1]{}
\newcommand{\g}{\mathsf{g}}                
\newcommand{\gsub}[1]{\mathsf{g}_{#1}}
\makeatletter
\newcommand{\panellabel}[2]{%
  \begingroup
  \edef\@currentlabel{\thefigure#1}%
  \label{#2}%
  \endgroup
}
\makeatother

\ShortHeadings{Guarantees for Stochastic Gradient Sampling}{Paulin and Whalley}

\begin{document}

\title{Theoretical guarantees for stochastic gradient sampling methods via Gaussian convolution inequalities}

\author{\name Daniel Paulin\hspace{0.3mm}\hyperlink{Hfootnote.1}{\footnotemark[1]} \email daniel.paulin@ntu.edu.sg\\
       \addr College of Computing and Data Science / School of Physical and Mathematical Sciences\\
       Nanyang Technological University, Singapore
       \AND
    \name Peter A. Whalley\hspace{0.3mm}\thanks{\protect\hypertarget{Hfootnote.1}{}Both authors contributed equally.} \email Peter.Whalley@warwick.ac.uk\\
       \addr Department of Statistics\\
       University of Warwick, Coventry CV4 7AL, United Kingdom
       }

\editor{}

\maketitle
\thispagestyle{plain} 

\begin{abstract}
We derive first-order (in the stepsize) bounds on the bias in Wasserstein distances of the invariant measure of stochastic gradient kinetic Langevin dynamics with minimal assumptions on the stochastic gradient noise. These bounds sharpen existing non-asymptotic guarantees for stochastic-gradient MCMC methods and provide a quantitative resolution of a previously open problem on invariant measure accuracy. The main technical ingredients are new Gaussian convolution inequalities controlling the Wasserstein-$p$ distance between a Gaussian convolved with a mean-zero perturbation and the Gaussian itself. We anticipate that these inequalities will be of independent interest beyond the present application. To complement the theory, we illustrate the benefits of stochastic gradient sampling methods on a one-dimensional model, a Bayesian logistic regression example and a Bayesian random-effect model including comparison with pseudo-marginal MCMC methods.
\end{abstract}

\begin{keywords}
   stochastic gradient MCMC; kinetic Langevin dynamics; Wasserstein distance; Gaussian convolution inequalities; Bayesian inference
\end{keywords}

\section{Introduction}
Sampling from a target probability distribution $\pi$ on $\mathbb{R}^d$ is a core computational task in modern statistics and applied probability, underpinning Bayesian inference, uncertainty quantification, and inverse problems \citep{Roberts2004,liu2001monte}. In many contemporary applications $\pi$ is high-dimensional and available only up to a normalising constant, making direct sampling infeasible. Markov chain Monte Carlo (MCMC) methods address this by simulating a Markov chain $(X_k)_{k\geq 0}$ whose distribution approaches $\pi$ (or an approximation of $\pi$) as $k\to\infty$. For such methods, two questions are central: how quickly the chain converges to its limiting measure, and how close that limiting measure is to the intended target $\pi$, as a function of the stepsize and problem parameters such as dimension and conditioning.

A widely used family of MCMC algorithms exploit first-order structure through the potential representation $\pi(dx)\propto \exp(-V(x))dx$ and the use of $\nabla V$. Such gradient-based schemes, including Langevin-type algorithms and their Metropolis-adjusted variants, often exhibit favourable scaling with dimension and conditioning, and are therefore a natural default in large-scale statistical computation \citep{roberts1998optimal,dalalyan2017theoretical,durmus2017nonasymptotic,che-book-2024+}.

The simplest gradient-based samplers are based on overdamped Langevin dynamics. Given a potential $V:\mathbb{R}^d\to\mathbb{R}$ and initialization $\mu\in\mathcal{P}(\mathbb{R}^d)$, overdamped Langevin dynamics is defined by the solution to 
\begin{equation}\label{eq:ovld}
    X_0\sim \mu, \quad dX_t = -\nabla V(X_t)dt + \sqrt{2}dW_t,
\end{equation}
where $(W_t)_{t\geq 0}$ is a standard $d$-dimensional Brownian motion. Under standard conditions on $V$, \eqref{eq:ovld} is ergodic with invariant measure having density proportional to $\exp(-V(x))$ \citep[see][]{pavliotis2014stochastic}, i.e. the target measure $\pi$. In practice one discretises \eqref{eq:ovld}; the resulting Markov chain is then used as an approximate sampler. The discrepancy between the invariant measure of the discretisation and $\pi$ is commonly referred to as the asymptotic bias \citep{DuEb24} or perfect sampling bias \citep{lemast2016}.

A recurring computational bottleneck, however, is the evaluation of $\nabla V$. In many statistical models $V$ decomposes as a sum of many contributions (for instance, a sum over observations), so that computing $\nabla V$ (and, for Metropolis corrections, evaluating $V$) can dominate the per-iteration cost. A standard remedy is to replace $\nabla V$ by an unbiased stochastic approximation $\mathcal{G}$, typically obtained by subsampling \citep{robbins1951stochastic}. This leads to stochastic-gradient MCMC methods, with stochastic gradient Langevin dynamics (SGLD) \citep{welling2011bayesian} as a canonical example.

In stochastic-gradient methods, approximation errors enter at two distinct levels. First, time discretisation perturbs the invariant distribution, so the discretised chain targets only an approximation to $\pi$. Second, replacing $\nabla V$ by an unbiased stochastic approximation introduces an additional perturbation in the invariant measure. We refer to the resulting contribution to the invariant measure error, and to the induced error in ergodic averages, as stochastic gradient bias. Understanding how this additional error scales with the stepsize and the stochastic-gradient noise is important for principled algorithm design and for assessing the reliability of stochastic-gradient MCMC in large-data settings.

In this paper, we focus on gradient-based samplers based on kinetic Langevin dynamics (also referred to as underdamped Langevin dynamics), which is defined on an extended state space to include a velocity variable. Given initialization $(X_{0},V_{0}) \sim \mu$, where $\mu \in \mathcal{P}(\mathbb{R}^{2d})$, kinetic Langevin dynamics is instead given by the solution to the following stochastic differential equation:
\begin{equation}\label{eq:kinetic_langevin}
    \begin{split}
    dX_{t} &= V_{t}dt,\\
    dV_{t} &= -\nabla V(X_{t}) dt - \gamma V_{t}dt + \sqrt{2\gamma}dW_{t},
\end{split}
\end{equation}
where $\gamma>0$ is a friction coefficient. Under weak assumptions, the unique invariant measure of the process $\{X_t,V_t\}_{t\geq 0}$ is of the form
\begin{equation}
\label{eq:inv}
\overline{\pi}(dxdv) \propto \exp\left(-V(x)-\frac{\|v\|^2}{2}\right)dxdv.
\end{equation}
Hence, marginally in $x$, one recovers the target measure of interest.

Kinetic Langevin dynamics forms the basis for many modern sampling methods \citep{brunger1984stochastic,leimkuhler2013rational,SGHMC} and often outperforms the overdamped diffusion \eqref{eq:ovld} in terms of convergence rate \citep[see][]{lu2022explicit,lu2026sharp}. As with overdamped Langevin, practical implementations discretise time, introducing asymptotic bias \citep{cheng2018underdamped,dalalyan2020sampling}; this is frequently ignored or corrected by Metropolization \citep{monmarche2021high}. Suitable discretisations of the kinetic dynamics can also yield smaller asymptotic bias than standard overdamped schemes, with more favourable dependence on key parameters \citep{cheng2018underdamped,dalalyan2020sampling,monmarche2021high,leimkuhler2023contractiona}. Due to the improved convergence rate in combination with computationally cheap integrators with reduced asymptotic bias we focus on gradient-based Monte Carlo methods based on \eqref{eq:kinetic_langevin}.

The aim of this paper is to characterise the additional bias induced by stochastic gradients in Wasserstein distance, establishing explicit $\mathcal{O}(h)$ Wasserstein bounds on the invariant measure of stochastic gradient kinetic Langevin dynamics (sometimes also referred to as a variant of stochastic gradient Hamiltonian Monte Carlo; \citealp{SGHMC}), under weak assumptions on the gradient noise and with explicit dependence on the dimension and the condition number.

We next recall the unadjusted Langevin algorithm (ULA), introduce the UBU integrator for \eqref{eq:kinetic_langevin} \citep{alamo2016technique,sanz2021wasserstein,skeel2002impulse}, and describe their stochastic-gradient counterparts, together with the assumptions used throughout.

\subsection{Discretised Langevin dynamics}\label{sec:integrators}

Throughout, $W_p$ denotes the Wasserstein-$p$ distance on $\mathcal{P}_p(\R^d)$, the probability measures with finite $p$-th moment, and we work under the following standard assumption.
\begin{assumption}\label{ass:strong-log-conc-log-smooth}
The potential $V\in C^2\tond*{\R^d}$ satisfies
$m I_d \preccurlyeq \hess V \preccurlyeq L I_d$
for some $0<m\leq L$. We denote by $\kappa \coloneqq \frac{L}{m}$ the condition number.
\end{assumption}

\paragraph{ULA.}
The Euler--Maruyama discretisation of \eqref{eq:ovld} with stepsize $h>0$ is
\begin{equation}\label{eq:ula}
    x_{k+1} = x_k - h\nabla V(x_k) + \sqrt{2h}\xi_{k+1},
\end{equation}
where $(\xi_k)_{k\geq 1}$ are i.i.d.\ $\mathcal{N}(0,I_d)$. Its invariant measure $\pi_h$ differs from $\pi$; under Assumption~\ref{ass:strong-log-conc-log-smooth} and $h\leq 1/L$, \citet{dur-maj-mia-2019} show that $W_2(\pi,\pi_h)\lesssim\sqrt{\kappa dh}$ \citep[see also][]{dalalyan2017theoretical,durmus2017nonasymptotic,durmus2019high,DuEb24}.

\paragraph{UBU discretisation of kinetic Langevin dynamics}
For kinetic Langevin dynamics, the integrator, stepsize, and friction $\gamma$ jointly determine both the convergence rate and the discretisation bias \citep{BoOw10,leimkuhler2013rational,gouraud2022hmc}. While one can apply an Euler-type scheme to \eqref{eq:kinetic_langevin} (as in SG-HMC of \citealp{SGHMC}), accurate second-order splitting integrators have been developed and analysed \citep[see][]{leimkuhler2013rational,sanz2021wasserstein,BUBthesis}, including BAOAB, ABOBA and OBABO \citep{OBABO,leimkuhler2013rational,lemast2016}. The UBU scheme \citep{BUBthesis,alamo2016technique,skeel1999integration,sanz2021wasserstein,ububu} requires one gradient evaluation per step and has strong order two. It splits \eqref{eq:kinetic_langevin} into
\[
\begin{pmatrix}
dX_{t} \\
dV_{t}
\end{pmatrix}
=
\underbrace{\begin{pmatrix}
0 \\
-\nabla V(X_{t})dt
\end{pmatrix}}_{\mathcal{B}}
+
\underbrace{\begin{pmatrix}
V_{t}dt \\
-\gamma V_{t} dt + \sqrt{2\gamma}dW_t
\end{pmatrix}}_{\mathcal{U}},
\]
where each sub-dynamics can be integrated exactly over a step of size $h>0$. Given $\gamma > 0$, let $\eta = \exp\left(-\gamma h/2\right)$, define the solution maps
\begin{equation}\label{eq:Bdef}
\mathcal{B}(x,v,h) = (x,v - h\nabla V(x)),
\end{equation}
and
\begin{align}\label{eq:Udef}
\mathcal{U}(x,v,h/2,\xi^{(1)},\xi^{(2)}) =\Big(x + \frac{1-\eta}{\gamma}&v+ \sqrt{\frac{2}{\gamma}}\left(\mathcal{Z}^{(1)}\left(h/2,\xi^{(1)}\right) - \mathcal{Z}^{(2)}\left(h/2,\xi^{(1)},\xi^{(2)}\right) \right),\\
\nonumber
&\quad \eta v + \sqrt{2\gamma}\mathcal{Z}^{(2)}\left(h/2,\xi^{(1)},\xi^{(2)}\right)\Big),
\end{align}
where
\begin{align}
\label{eq:Z1def}&\mathcal{Z}^{(1)}\left(h/2,\xi^{(1)}\right) = \sqrt{\frac{h}{2}}\xi^{(1)},\\
\label{eq:Z2def}&\mathcal{Z}^{(2)}\left(h/2,\xi^{(1)},\xi^{(2)}\right) = \quad \sqrt{\frac{1-\eta^{2}}{2\gamma}}\Bigg(\sqrt{\frac{1-\eta}{1+\eta}\cdot \frac{4}{\gamma h}}\xi^{(1)} + \sqrt{1-\frac{1-\eta}{1+\eta}\cdot\frac{4}{\gamma h}}\xi^{(2)}\Bigg),
\end{align}
where $\xi^{(1)},\xi^{(2)} \sim \mathcal{N}(0_{d},I_{d})$ are independent $d$-dimensional standard Gaussians.
The UBU scheme applies a half-step of $\mathcal{U}$ (with $h$ replaced by $h/2$), then one $\mathcal{B}$ step, followed by another half-step of $\mathcal{U}$. Its symmetry cancels the leading-order term in the error expansion (in the absence of gradient noise) \citep{alamo2016technique,sanz2021wasserstein}. With
\begin{equation}\label{eq:EFdef}
\mathcal{E}(t)=e^{-\gamma t},
\qquad
\mathcal{F}(t)=\frac{1-e^{-\gamma t}}{\gamma},
\end{equation}
the UBU update with stepsize $h$ can equivalently be written as
\begin{equation}\label{eq:disc}
\begin{split}
y_{k} &= x_{k} + \mathcal{F}(h/2)v_{k} + \sqrt{2\gamma}\int^{(k+1/2)h}_{kh}\mathcal{F}((k+1/2)h-s)dW_s,\\
v_{k+1} &= \mathcal{E}(h)v_k - h\mathcal{E}(h/2)g_{k} + \sqrt{2\gamma}\int^{(k+1)h}_{kh}\mathcal{E}((k+1)h-s)dW_s,\\
x_{k+1} &= x_{k} + \mathcal{F}(h)v_{k} - h\mathcal{F}(h/2)g_{k} + \sqrt{2\gamma}\int^{(k+1)h}_{kh}\mathcal{F}((k+1)h-s)dW_s,
\end{split}
\end{equation}
with $g_k=\nabla V(y_k)$; here $y_k$ is the position after the first $\mathcal U$ half-step, and the formula follows from the variation-of-constants representation of the $\mathcal U$ dynamics \citep[cf.][]{sanz2021wasserstein}.

\paragraph{Stochastic gradient Langevin dynamics}
As noted above, computing $\nabla V$ can be prohibitively expensive (or sometimes not possible), and it is often advantageous to replace it in \eqref{eq:ula} by a cheaper unbiased estimator. The resulting algorithm is stochastic gradient Langevin dynamics (SGLD). We introduce and follow the stochastic-gradient setup of \cite{leimkuhler2023contractionb}.

\begin{definition}\label{def:stochastic_gradient}
A \textit{stochastic gradient approximation} of a potential $V$ is specified by a measurable map $\mathcal{G}:\R^d \times \Omega \to \R^d$ and a probability distribution $\Gamma$ on a Polish space $\Omega$ such that for every $x\in \R^d$ and $\omega\sim \Gamma$,
\[
\E(\mathcal{G}(x,\omega)) = \nabla V(x).
\]
We denote the stochastic gradient by $(\mathcal{G}, \Gamma)$.
\end{definition}

\begin{assumption}\label{Assumption:Bounded_Variance}
Assume that for $\Gamma$-almost every $\omega$, the map
$x\mapsto \mathcal{G}(x,\omega)$ is $C^1$ on $\mathbb{R}^d$.
Moreover, there exists $C_{G}\geq 0$ such that
\[
\sup_{x\in \mathbb{R}^{d}}
\E_{\omega\sim\Gamma}\left\|D_{x}\mathcal{G}(x,\omega) - \nabla^{2}V(x)\right\|_{\op}^{2}
\leq C_{G}.
\]
Finally, assume that there exists $x_{\star}\in\mathbb{R}^{d}$ such that
\[
\E_{\omega\sim\Gamma}\|\mathcal{G}(x_{\star},\omega)\|^{2}<\infty.
\]
\end{assumption}

Throughout, when $C_G=0$, a stepsize restriction containing $1/C_G$ is interpreted as vacuous (equivalently, the corresponding upper bound is $+\infty$).

\begin{assumption}\label{assum:sg_variance}
Let $\pi$ denote the $x$-marginal of the invariant law $\overline{\pi}$ in \eqref{eq:inv}, then for $p\geq 1$, define
\[
\sigma_{p} := \left(\E_{X \sim \pi,\ \omega\sim \Gamma}\|\mathcal{G}(X,\omega) - \nabla V(X)\|^{p}\right)^{1/p}.
\]
Assume $\sigma_p<\infty$ for the values of $p$ used below.
\end{assumption}

We define SGLD by replacing $\nabla V(x_k)$ in \eqref{eq:ula} with $\mathcal{G}(x_k,\omega_{k+1})$ for i.i.d.\ $\omega_{k+1}\sim \Gamma$. Similarly, stochastic-gradient UBU (SG-UBU) uses $g_k=\mathcal{G}(y_k,\omega_{k+1})$ in \eqref{eq:disc}. 

For SGLD, under Assumption~\ref{ass:strong-log-conc-log-smooth} and assumptions on the gradient noise, \cite{DalalyanSG} show that
\begin{equation}\label{eq:sgld_bias_basic}
W_2(\pi_h,\pi)
\ \lesssim\
\sqrt{h}\frac{L\sqrt d}{m}
+
\sigma_2\sqrt{\frac{h}{m}}.
\end{equation}
If $\nabla^2V$ is in addition Lipschitz, the first term improves to $\mathcal{O}(h)$, but the stochastic-gradient term remains of order $h^{1/2}$. A similar $h^{1/2}$ rate has been shown for stochastic gradient kinetic Langevin dynamics, and more generally for stochastic gradient generalised Hamiltonian Monte Carlo \citep[see][]{gouraud2022hmc}.

In numerical experiments (as illustrated by examples in Section~\ref{sec:simulations}), one typically observes a bias that decays at least linearly in the stepsize, and not like $h^{1/2}$ as indicated by these theoretical guarantees. Additionally, for sufficiently smooth test functions $\phi$, \cite{vollmer2016exploration} prove $\mathcal{O}(h)$ weak (ergodic-average) bias for SGLD-type schemes, with finite-time bias/variance guarantees. However, these approaches, unlike Wasserstein-based analyses, do not bound the invariant measure directly and do not give non-asymptotic guarantees with explicit dependence on key parameters such as dimension and conditioning.

For overdamped SGLD, first-order invariant-measure bounds are available under stronger assumptions. \citet{li2025sharp} prove a uniform-in-time $\mathcal{O}(h^2)$ bound on the relative entropy between SGLD and the Langevin diffusion, which yields $\mathcal{O}(h)$ bounds on the distance between their invariant measures, assuming in particular bounded stochastic-gradient errors and additional smoothness. In a companion paper \citep{paulin2026sgld}, we prove $W_2$ bias bounds for SGLD of order $h$ under a fourth-moment assumption on the gradient error and of order $h^{1/2}$ under only a second-moment assumption. The present paper treats the kinetic dynamics, whose UBU discretisation has a second-order full-gradient bias, under moment, tail or Poincar\'e-type assumptions on the noise, and with explicit constants.

Closest to the present work, \cite{lu2025mean} establish first-order, $\mathcal{O}(h)$, bias for the same stochastic gradient UBU scheme, as well as for its variance-reduced SVRG-UBU and SAGA-UBU counterparts, under a log-concavity assumption. Their analysis is of a different nature to ours: it controls the mean square error of finite-time ergodic averages of sufficiently smooth test functions, rather than the invariant measure itself. The two sets of results are therefore complementary, with different strengths. On the one hand, \cite{lu2025mean} require only log-concavity, whereas we work under the stronger assumption of strong log-concavity and log-smoothness (Assumption~\ref{ass:strong-log-conc-log-smooth}). On the other hand, our Wasserstein bounds control the invariant measure $\overline{\pi}_{h}$ directly, and hence the bias of every Lipschitz observable, with explicit dependence on the dimension $d$ and the condition number $\kappa$; moreover, they require fewer derivatives, namely $V \in C^{2}$ together with a single-derivative (bounded Jacobian variance) condition on the stochastic gradient and a finite second moment at one reference point (Assumption~\ref{Assumption:Bounded_Variance}), and in particular no smoothness of the test functions or of the solution of the associated Poisson equation, as is typically needed in weak-error and mean-square-error expansions. Finally, the transition from first-order to second-order accuracy that \cite{lu2025mean} establishes for variance-reduced gradients is closely related to the $\mathcal{O}(h^{2})$ behaviour under epoch-wise without-replacement subsampling that we discuss in the conclusion (Section~\ref{sec:conclusion}).

\section{Main results}

Our main results have two parts. First, we establish Gaussian convolution inequalities that control
$W_p(\mu*\g,\g)$ for centred perturbations $\mu$. Second, we use these inequalities to quantify
the asymptotic bias of stochastic-gradient UBU, obtaining first-order accuracy of its invariant measure in Wasserstein distance, in agreement with practice. When combined with the subsequently stated Wasserstein convergence rates for the stochastic gradient UBU schemes these provide state-of-the-art non-asymptotic guarantees for an MCMC method which utilises Robbins-Monro  stochastic gradient subsampling \citep[see][]{robbins1951stochastic}.

\subsection{Gaussian convolution inequalities}

We first state several Gaussian convolution inequalities under a varying set of assumptions.

\subsubsection{General \texorpdfstring{$p$}{p}-Wasserstein inequality}
We begin with a Gaussian convolution inequality for centred measures, valid for all $p\geq 1$.

\begin{restatable}{theorem}{convGeneral}\label{thm:conv_general}
Fix $d\geq 1$ and $p\geq 1$. Let $\mu \in \mathcal{P}(\mathbb{R}^d)$ satisfy
\begin{equation}\label{eq:centred}
\int_{\mathbb{R}^d} x\mu(dx)=0
\qquad\text{and}\qquad
\int_{\mathbb{R}^d}\|x\|^{2p}\mu(dx)<\infty.
\end{equation}
Let $\g=\mathcal{N}(0, I_d)$, then
\begin{equation}\label{eq:general_mu_convolution_bound_main}
W_p(\mu*\g,\g)
\le
\frac{K_p}{1-(1-2^{-2p})^{1/p}}
\left(\int_{\mathbb{R}^d}\|x\|^{2p}\mu(dx)\right)^{1/p},
\end{equation}
where $K_p = \max\left\{1, \frac{C_p}{2} + \frac{1}{3}\right\}$ and
$C_p = \big(\mathbb{E}_{\xi \sim \mathcal{N}(0,1)}|\xi|^p\big)^{1/p}$.
\end{restatable}
\begin{remark}
    Under the assumptions of Theorem~\ref{thm:conv_general} one can rescale the resulting inequality to achieve
    \begin{equation}\label{eq:s-scaling}
    W_p(\mu*\g_{s},\g_{s})
    \le
    \frac{1}{\sqrt{s}} \cdot\frac{K_p}{\left(1-(1-2^{-2p})^{1/p}\right)}
    \left(\int_{\mathbb{R}^d}\|x\|^{2p}\mu(dx)\right)^{1/p},
    \end{equation}
    where $\gsub{s}:=\mathcal N(0,s I_d)$, for $s > 0$. Consequently as $s \to \infty$, $ W_p(\mu*\g_{s},\g_{s}) \to 0$, which is not the case under the inequality $ W_p(\mu*\g_{s},\g_{s}) \leq \left(\int_{\mathbb{R}^d}\|x\|^{p}\mu(dx)\right)^{1/p}$ achieved under trivial coupling. The improved factor of $s^{-1/2}$ is the key ingredient to achieve asymptotic bias results for stochastic gradient MCMC of order $h$ with respect to the stepsize $h > 0$, as opposed to order $h^{1/2}$ using a trivial coupling of the stochastic gradient noise. 
    
    We also remark that the $s^{-\frac{1}{2}}$ scaling in \eqref{eq:s-scaling} was shown asymptotically in \cite{Chen-N-W-22} and used in \cite{beyler2025convergence} within the context of studying convergence of diffusion models. \cite{srinivasan2025poisson} also used a one-dimensional version of this type of inequality for bounded random variables to bypass the strong order barrier for a randomised midpoint method.
\end{remark}
\begin{remark}
In Section~\ref{app:stochastic-localization-wp} of the Appendix, we show a similar bound with a sharper constant for $p\geq 2$ via a stochastic localization argument (Theorem~\ref{thm:stochastic-localization-wp}), which we use for $p=2$ in Corollary~\ref{cor:plug_in}.
\end{remark}

For $p=1$, a short Gaussian-smoothing argument gives a sharper bound with the optimal constant.

\begin{theorem}\label{thm:conv_W1}
Let $\mu\in\mathcal P(\R^d)$ be centred with $\int\|x\|^2\mu(dx)<\infty$, and let $\g=\mathcal N(0,I_d)$. Then
\[
W_1(\mu*\g,\g)\leq\frac{1}{\sqrt{2\pi}}\int_{\R^d}\|x\|^2\mu(dx),
\]
and the constant $1/\sqrt{2\pi}$ cannot be improved.
\end{theorem}

\begin{proof}
Let $X\sim\mu$ and $Z\sim\g$ be independent, let $f$ be $1$-Lipschitz, and set $F(x):=\E f(x+Z)$. Gaussian integration by parts gives $\nabla^2F(x)=\E[\nabla f(x+Z)Z^\top]$, so for unit vectors $u,v$, $|u^\top\nabla^2F(x)v|\leq\E|\langle Z,v\rangle|=\sqrt{2/\pi}$. By Taylor's formula and $\E X=0$,
\[
|\E f(X+Z)-\E f(Z)|=|\E F(X)-F(0)-\langle\nabla F(0),\E X\rangle|\leq\frac12\sqrt{\frac2\pi}\,\E\|X\|^2,
\]
and the bound follows from Kantorovich--Rubinstein duality. For sharpness, take $d=1$, $X=\pm\delta$ with equal probabilities and $f(x)=|x|$: then $\E|X+Z|-\E|Z|=\delta(2\Phi(\delta)-1)+2\phi(\delta)-2\phi(0)=\delta^2/\sqrt{2\pi}+\mathcal{O}(\delta^4)$, where $\phi$ and $\Phi$ are the standard normal density and distribution function.
\end{proof}

For $p=1$, Theorem~\ref{thm:conv_W1} improves the constant $4$ of Theorem~\ref{thm:conv_general} by a factor of about ten.
\subsubsection{A refined \texorpdfstring{$\{1,2\}$}{1,2}-Wasserstein inequality}
For $p\in\{1,2\}$ we also obtain a refinement that separates a tail term from a (truncated) covariance term.

\begin{restatable}{theorem}{convRefined}\label{thm:conv_refined_12}
Let $X\sim\mu$ satisfy \eqref{eq:centred} for $p=1$ and
define the tail quantities at threshold $1$ by
\[
\tau_1 := \mathbb{E}\big[\|X\|\mathbf 1_{\{\|X\|>1\}}\big],\qquad
\tau_2 := \Big(\mathbb{E}\big[\|X\|^2\mathbf 1_{\{\|X\|>1\}}\big]\Big)^{1/2},
\]
and let
\[
\widetilde X:=X\mathbf 1_{\{\|X\|\leq 1\}}-\mathbb{E}[X\mathbf 1_{\{\|X\|\leq 1\}}],
\qquad
\widetilde\Sigma:=\mathrm{Cov}(\widetilde X).
\]
Let $c_4=(e^4-5)/16$, then
\[
W_1(\mu*\g,\g)
\le
\Big(\tau_1+\min(1,\tau_1)\Big)
+
\sqrt{2\log\left(1+c_4\|\widetilde\Sigma\|_F^2\right)},
\]
and
\[
W_2(\mu*\g,\g)
\le
\Big(\tau_2+\min(1,\tau_2)\Big)
+
\sqrt{2\log\left(1+c_4\|\widetilde\Sigma\|_F^2\right)}.
\]
\end{restatable}

\begin{remark}\label{rem:spike}
In Section~\ref{sec:spike} we give a spike example showing that, in general, a bound purely in terms
of $\|\Sigma\|_F$ cannot hold without additional tail control.
\end{remark}

\subsubsection{\texorpdfstring{$\{1,2\}$}{1,2}-Wasserstein and Kullback--Leibler inequalities under a Poincar\'e inequality}
\begin{definition}\label{def:poincare_constant}
For a measure $\mu$ we denote by $C_P(\mu)$ its Poincar\'e constant, i.e.\ the smallest $C\in[0,\infty]$
such that
\begin{equation}\label{eq:poincare}
\mathrm{Var}_\mu(f)\leq C \int \|\nabla f\|^2d\mu
\quad \text{for all smooth } f\in L^2(\mu).
\end{equation}
\end{definition}

\begin{restatable}{theorem}{convPoincare}\label{thm:conv_poincare}
For every centred $\mu$, which is absolutely continuous with respect to the Lebesgue measure, has finite covariance $\Sigma$ and Poincar\'e constant $C_{P}$,
\[
W_1(\mu*\g,\g)\ \le\ W_2(\mu*\g,\g)\ \le\ \frac12\sqrt{C_P\mathrm{tr}(\Sigma)}.
\]
\end{restatable}

\begin{restatable}{theorem}{convPoincareKL}\label{thm:conv_poincare_KL}
Under the assumptions of Theorem~\ref{thm:conv_poincare},
\[
\mathrm{KL}(\mu*\g\|\g)\ \le\ \frac14\ C_P\mathrm{tr}(\Sigma),
\]
where $\mathrm{KL}(\rho\|\g):=\int\log\frac{d\rho}{d\g}d\rho$ denotes the relative entropy (Kullback--Leibler divergence).
\end{restatable}

Theorems~\ref{thm:conv_poincare} and~\ref{thm:conv_poincare_KL} are proved in Appendix~\ref{sec:poincare-appendix} (Theorem~\ref{thm:poincare-KL}) in a sharper Stein-kernel form, and both constants are attained in the small-noise limit for Gaussian $\mu$ (Remark~\ref{rem:KL-sharp}).

\subsection{Convergence of stochastic gradient UBU}
We summarise Wasserstein convergence results from \cite{ububu}, which uses the approach of
\cite{leimkuhler2023contractiona}. Related coupling approaches for kinetic Langevin dynamics and its
discretisations include, among others, \cite{monmarche2021high,sanz2021wasserstein,chakreflection,schuh2024convergence}
(discretised) and \cite{eberle2019couplings,cheng2018underdamped,dalalyan2020sampling,schuh2022global} (continuous). Wasserstein convergence for the UBU discretisation was first studied in \cite{sanz2021wasserstein}. Because the generator is hypoelliptic, a contraction estimate of this form with prefactor one generally fails in the unmodified Euclidean norm. Instead, we consider a ``twisted'' Euclidean norm which is equivalent to the standard norm up to a constant, so that Wasserstein convergence in the Euclidean norm follows \citep[see][for the same setup]{monmarche2021high}.

\begin{definition}\label{def:weightednorm}
For $z = (x,v) \in \R^{2d}$ define
\[
\| z \|^{2}_{a,b} = \| x \|^{2} + 2b \langle x,v \rangle + a \| v \|^{2},
\]
for $a,b > 0$ with $b^{2}<a$.
\end{definition}

\begin{remark}\label{rem:norm_equiv}
If $b^2<a/4$, then $\|\cdot\|_{a,b}$ is equivalent to the Euclidean norm on $\R^{2d}$, and
\begin{equation}\label{eq:normequiv}
\frac{1}{2}\min(a,1) \|z\|^2\leq \frac{1}{2}\|z\|^{2}_{a,0} \leq \|z\|^{2}_{a,b} \leq \frac{3}{2}\|z\|^{2}_{a,0}\leq \frac{3}{2}\max(a,1) \|z\|^2.
\end{equation}
\end{remark}

\begin{definition}\label{def:wass}
For $p \in [1,\infty)$, and $\mu,\nu \in \mathcal{P}_p(\R^{2d})$, the $p$-Wasserstein distance associated with $\|\cdot \|_{a,b}$ between $\mu$ and $\nu$ is
\begin{equation}\label{eq:wass_dist}
W_{p,a,b}(\mu,\nu) =
\Big(\inf_{\tau \in \Pi(\mu,\nu)} \int_{\R^{2d} \times \R^{2d}} \|z_1 - z_2\|^p_{a,b}d\tau(z_1,z_2)\Big)^{1/p}.
\end{equation}
\end{definition}

\begin{restatable}{proposition}{ubuWasserstein}\label{prop:Wasserstein}
Suppose that $V$ satisfies Assumption~\ref{ass:strong-log-conc-log-smooth} and let
\begin{equation}\label{eq:abeq}
a=\frac{1}{L},\quad b=\frac{1}{\gamma}.
\end{equation}
Let $P_{h}$ denote the transition kernel of one UBU step with stepsize $h$, then for all $\gamma \geq \sqrt{8L}$,
$0< h < \frac{1}{2\gamma}$, $p \geq 1$, and $\mu,\nu \in \mathcal{P}_{p}(\R^{2d})$,
\[
W_{p,a,b}\left(\nu P^n_{h} ,\mu P^n_{h} \right)
\le
\left(1 - \frac{mh}{4\gamma}\right)^{n/2} W_{p,a,b}\left(\nu,\mu\right),
\qquad n\in\N.
\]
Moreover, $P_{h}$ has a unique invariant measure $\overline{\pi}_{h}$ with $\overline{\pi}_{h} \in \mathcal{P}_{p}(\R^{2d})$.
\end{restatable}
\begin{proof}
    See \cite[Proposition D.6]{ububu}, whose proof follows \cite[Corollary 20]{monmarche2021high} and \cite{leimkuhler2023contractiona}.
\end{proof}
\begin{restatable}{proposition}{sgubuWasserstein}\label{prop:wass_SGUBU}
Consider UBU with stochastic gradients and transition kernel $P_{h}$, where $V$ satisfies Assumption~\ref{ass:strong-log-conc-log-smooth}.
Assume a stochastic gradient $(\mathcal{G},\Gamma)$ (as in Definition~\ref{def:stochastic_gradient}) satisfying Assumption~\ref{Assumption:Bounded_Variance} with constant $C_{G}$.
Consider $\mu,\nu \in \mathcal{P}_{2}(\R^{2d})$ and any two synchronously coupled stochastic gradient UBU chains $(x_{k},v_{k})_{k\in \mathbb{N}}$ and $(\Tilde{x}_{k},\Tilde{v}_{k})_{k\in \mathbb{N}}$ initialized at $\mu$ and $\nu$ respectively. Then under the same assumptions as in Proposition~\ref{prop:Wasserstein}, we have for all $n \in \mathbb{N}$,
\begin{align*}
    &\E\left[\|(x_{n+1} - \Tilde{x}_{n+1},v_{n+1} - \Tilde{v}_{n+1})\|^{2}_{a,b} \mid x_n, \Tilde{x}_{n},v_{n},\Tilde{v}_{n}\right]\\
    &\quad\le
    \left(1-\frac{mh}{4\gamma} + \frac{8h^{2}C_{G}}{7L}\right)
    \|(x_{n} - \Tilde{x}_{n},v_{n} - \Tilde{v}_{n})\|^{2}_{a,b}.
\end{align*}
Consequently, for $p \in \{1,2\}$,
\[
W_{p,a,b}\left(\nu P^n_{h} ,\mu P^n_{h} \right)
\le
\left(1 - \frac{mh}{4\gamma} + \frac{8h^{2}C_{G}}{7L}\right)^{n/2}
W_{p,a,b}\left(\nu,\mu\right).
\]
Moreover, provided $h<\min\left\{\frac{1}{2\gamma},\frac{7mL}{32C_{G}\gamma}\right\}$, SG-UBU has a unique invariant measure $\overline{\pi}_{h}\in \mathcal{P}_{2}(\R^{2d})$.
\end{restatable}

\begin{proof}
See Appendix~\ref{app:ubu_conv}.
\end{proof}

\subsection{Asymptotic bias of stochastic gradient UBU}

We now state the new asymptotic bias estimates for stochastic gradient UBU, which are then combined with the Gaussian convolution inequalities under varying assumptions on the stochastic gradient noise.

\begin{restatable}{theorem}{sgubuBias}\label{theorem:bias_sgUBU}
Consider the stochastic gradient UBU scheme under Assumption~\ref{assum:sg_variance} and the same assumptions as Proposition~\ref{prop:wass_SGUBU}. For
\[
h<\min\left\{\frac{1}{2\gamma},\frac{7mL}{32C_{G}\gamma}\right\},
\qquad p \in \{1,2\},
\]
we have
\begin{equation}\label{eq:sg_ubu_bias}
\begin{split}
W_{p,a,b}(\overline{\pi},\overline{\pi}_{h})
&\le
\frac{\gamma L h}{mL - \frac{32}{7}hC_{G}\gamma}\left[\frac{24}{7}\sqrt{d}\left(\sqrt{L} + \gamma\right) + \frac{4C^{1/2}_{G}\sqrt{d}}{\sqrt{L}}\right.\\
&\qquad\left.{} + 4\sigma_{p} + \frac{9(\mathbb{E}_{x\sim \pi}W^{p}_{p}(\mu_{x} * \g, \g))^{1/p}}{h^{2}\sqrt{L}}\right],
\end{split}
\end{equation}
where $he^{-h\gamma/2}(\nabla V(x) - \mathcal{G}(x,\omega)) \sim \mu_{x}$ for $\omega\sim \Gamma$.
\end{restatable}
In \eqref{eq:sg_ubu_bias}, the prefactor is $\mathcal{O}(h)$. The bracketed terms separate the full-gradient UBU discretisation error from the additional stochastic-gradient error, with the last term controlled via our Gaussian convolution inequalities. Hence $W_{p,a,b}(\overline{\pi},\overline{\pi}_h)=\mathcal{O}(h)$ under mild assumptions on the gradient noise. Appendix~\ref{sec:gaussian-sharpness} computes the invariant law exactly for a standard Gaussian target with additive Gaussian stochastic-gradient noise and shows that the $h\sqrt d$ rate is sharp up to constant factors in that setting. More precisely, we can combine Theorem~\ref{theorem:bias_sgUBU} with the following Gaussian convolution inequalities which can easily be substituted into \eqref{eq:sg_ubu_bias}.

\begin{restatable}{corollary}{pluginBounds}\label{cor:plug_in}
For $p \in \{1,2\}$, let $X \sim \pi$, $\omega \sim \Gamma$, and set $\mathcal{R}(X,\omega)=\nabla V(X)-\mathcal{G}(X,\omega)$.
\begin{enumerate}[label=(\roman*),leftmargin=2.3em]
\item If $\mathcal{R}(X,\omega)$ has finite $2p$-th moment (equivalently $\sigma_{2p}<\infty$), then in \eqref{eq:sg_ubu_bias},
\[
\frac{9(\mathbb{E}_{x\sim \pi}W^{p}_{p}(\mu_{x} * \g, \g))^{1/p}}{h^{2}\sqrt{L}}
\le
\frac{c_p\sigma_{2p}^2}{\sqrt{L}},
\qquad c_1=\frac{9}{\sqrt{2\pi}}<3.6,\quad c_2=\frac92.
\]
In particular, the $W_1$ bound requires only $\sigma_2<\infty$, whereas the $W_2$ bound requires $\sigma_4<\infty$.

\item If $\mathcal{R}(x,\omega)$ is absolutely continuous with respect to the Lebesgue measure and satisfies a Poincar\'e inequality  for each $x \in \mathbb{R}^{d}$ with constant $C_{P}(x)$, then
\[
\frac{9(\mathbb{E}_{x\sim \pi}W^{p}_{p}(\mu_{x} * \g, \g))^{1/p}}{h^{2}\sqrt{L}}
\le
\frac{9\sqrt{\int C_{P}(x)\mathrm{tr}\left(\mathrm{Cov}(\mathcal{R}(x,\cdot))\right)d\pi(x)}}{2\sqrt{L}}.
\]

\item If $p=1$ and only second moments are assumed, then, with
\[
Y:=he^{-h\gamma/2}\mathcal{R}(X,\omega),
\qquad
\tau_1(Y):=\mathbb{E}\left[\|Y\|\mathbf 1_{\{\|Y\|>1\}}\right],
\]
we have the loose bound
\[
\frac{9\mathbb{E}_{x\sim \pi}W_{1}(\mu_{x} * \g, \g)}{h^{2}\sqrt{L}}
\leq
\frac{18\tau_1(Y)}{h^{2}\sqrt{L}}
+
\frac{23e^{-h\gamma}\sqrt{d}}{\sqrt{L}}
\mathbb{E}_{X\sim\pi}\!\left[
\lambda_{\max}\left(\mathrm{Cov}(\mathcal{R}(X,\cdot))\right)
\right].
\]
The final expectation is finite under the stated second-moment assumption, since $\lambda_{\max}(\mathrm{Cov}(\mathcal R(X,\cdot)))\leq\mathbb{E}[\|\mathcal R(X,\omega)\|^2\mid X]$.

\item If $p=2$, assume in addition that
\[
\mathbb{E}_{X\sim\pi}\left[
\lambda_{\max}\left(\mathrm{Cov}(\mathcal{R}(X,\cdot))\right)^2
\right]<\infty.
\]
With
\[
Y:=he^{-h\gamma/2}\mathcal{R}(X,\omega),
\qquad
\tau_2(Y):=\left(\mathbb{E}\left[\|Y\|^2\mathbf 1_{\{\|Y\|>1\}}\right]\right)^{1/2},
\]
we have
\[
\frac{9(\mathbb{E}_{x\sim \pi}W^{2}_{2}(\mu_{x} * \g, \g))^{1/2}}{h^{2}\sqrt{L}} \leq \frac{18\tau_2(Y)}{h^{2}\sqrt{L}} + \frac{23e^{-h\gamma}\sqrt{d}}{\sqrt{L}}
\left(\mathbb{E}_{X\sim\pi}\left[\lambda_{\max}\left(\mathrm{Cov}(\mathcal{R}(X,\cdot))\right)^2\right]\right)^{1/2}.
\]
\end{enumerate}
\end{restatable}
\begin{remark}
If $C_P(x)$ and $\|\mathrm{Cov}(\mathcal{R}(x,\cdot))\|$ are independent of $d$ and $\sigma_{1}=\mathcal{O}(\sqrt{d})$, part (ii) yields a Wasserstein-1 bound of $\mathcal{O}(h\sqrt{d})$; the same holds under part (iii) for bounded noise, $h e^{-h\gamma/2}\|\mathcal{R}(x,\omega)\|\leq 1$. Under only the second-moment assumption (i) with $p=1$, the bound contains the term $h\sigma_{2}^2=h\E_{X\sim\pi}[\mathrm{Tr}(\mathrm{Cov}(\mathcal{R}(X,\cdot)))]$, which is typically $\mathcal{O}(hd)$. Thus the tail behaviour of the noise affects our bounds. The $h\sqrt d$ rate is sharp in an exactly solvable Gaussian example (Appendix~\ref{sec:gaussian-sharpness}), while Section~\ref{sec:spike:SGUBU} gives heuristic and numerical evidence of worse dimension dependence for heavy-tailed spike noise.
\end{remark}

\begin{corollary}\label{cor:minibatch}
Suppose that the assumptions of Theorem \ref{theorem:bias_sgUBU} hold and that
\[
V(x)=V_0(x)+\sum_{i=1}^{N}V_i(x),
\qquad V_i\in C^2(\mathbb{R}^d),
\]
then write
\[
g_i(x)=\nabla V_i(x),\quad H_i(x)=\nabla^2V_i(x),\quad
\bar g(x)=\frac1N\sum_{i=1}^N g_i(x),\quad
\bar H(x)=\frac1N\sum_{i=1}^N H_i(x).
\]
For a minibatch of size $B$, consider the unbiased estimator
\[
\mathcal G_B(x)=\nabla V_0(x)+\frac{N}{B}\sum_{j=1}^{B}g_{I_j}(x),
\]
where either $I_1,\ldots,I_B$ are sampled independently and uniformly from $\{1,\ldots,N\}$, or the minibatch is a uniformly sampled subset of size $B$. Define
\[
\rho_B=
\begin{cases}
B^{-1},&\text{sampling with replacement},\\[1mm]
\dfrac{N-B}{B(N-1)},&\text{sampling without replacement},
\end{cases}
\]
and assume
\[
\mathsf H^2:=\sup_{x\in\mathbb{R}^d}\frac1N\sum_{i=1}^N
\|H_i(x)-\bar H(x)\|_F^2<\infty,
\qquad
\mathsf S^2:=\int\frac1N\sum_{i=1}^N
\|g_i(x)-\bar g(x)\|^2\,\pi(dx)<\infty.
\]
Then Assumption~\ref{Assumption:Bounded_Variance} holds with $C_G\leq N^2\rho_B\mathsf H^2$ , $\sigma_{1}\leq\sigma_2=N\sqrt{\rho_B}\,\mathsf{S}$, and if
$
h<\min\left\{\frac1{2\gamma},
\frac{7mL}{32\gamma N^2\rho_B\mathsf H^2}\right\},
$
then
\begin{align}
W_{1,a,b}(\overline\pi,\overline\pi_h)
&\leq
\frac{\gamma Lh}{mL-\frac{32}{7}h\gamma N^2\rho_B\mathsf H^2}
\Bigg[
\frac{24}{7}\sqrt d(\sqrt L+\gamma)
+\frac{4N\sqrt{\rho_B}\,\mathsf H\sqrt d}{\sqrt L}
\nonumber\\[-1mm]
&\hspace{44mm}
+4N\sqrt{\rho_B}\,\mathsf S
+\frac{9N^2\rho_B\mathsf S^2}{\sqrt{2\pi L}}
\Bigg].
\label{eq:minibatch-corollary}
\end{align}
\end{corollary}
\begin{remark}
For sampling with replacement, the stochastic-gradient terms therefore scale explicitly as $N/\sqrt B$ and $N^2/B$ under the sum convention for $V$ used above. For a full batch without replacement, $\rho_B=0$ and the stochastic-gradient terms vanish.
\end{remark}

\begin{proof}
Unbiasedness is immediate, and we have
\[
\mathbb{E}\|\mathcal G_B(x)-\nabla V(x)\|^2
=N^2\rho_B\frac1N\sum_{i=1}^N\|g_i(x)-\bar g(x)\|^2.
\]
The same identity in the Frobenius inner product, together with $\|A\|_{\mathrm{op}}\leq\|A\|_F$, yields
\begin{align*}
\mathbb{E}\|D\mathcal G_B(x)-\nabla^2V(x)\|_{\mathrm{op}}^2
&\leq
\mathbb{E}\|D\mathcal G_B(x)-\nabla^2V(x)\|_F^2\\
&=N^2\rho_B\frac1N\sum_{i=1}^N\|H_i(x)-\bar H(x)\|_F^2.
\end{align*}
Since the finite sum consists of deterministic $C^2$ functions, the reference-point moment condition in Assumption~\ref{Assumption:Bounded_Variance} is automatic. Integrating the first identity against $\pi$ gives the formula for $\sigma_2$. Finally, apply Theorem~\ref{theorem:bias_sgUBU} with $p=1$; by Corollary~\ref{cor:plug_in}(i), the convolution term is bounded by $9\sigma_2^2/\sqrt{2\pi L}$. Substituting this in gives \eqref{eq:minibatch-corollary}.
\end{proof}

\section{Simulations}\label{sec:simulations}

All the code for the simulations in this paper is available at \texttt{\url{https://github.com/PAWhalley/SG-UBU}}.

\subsection{Numerical evaluation of integrator bias on toy example}
\label{sec:numevalbias}
To illustrate the asymptotic bias of the stochastic gradient methods, we consider a simple one-dimensional Gaussian target of the form $U(x)=U_1(x)+U_2(x)$, where $U_i(x)=\frac{(x-x_i)^2}{\sigma_i^2}$ are quadratic potentials. We take $x_1=-1$, $x_2=1$, $\sigma_1=0.5$, and $\sigma_2=2$. With batch size $1$ there are two batches, so the stochastic gradient is obtained by sampling one of the two terms uniformly at each iteration and rescaling to maintain unbiasedness.

Since the target distribution is Gaussian, we estimate the Wasserstein-$1$ bias by comparing samples from the stationary distribution of each Markov chain to i.i.d.\ samples from the exact target. Concretely, for a given method and stepsize $h$, let $(X_k)_{k\geq 1}$ denote the $x$-marginal of the Markov chain after burn-in and let $(Y_k)_{k\geq 1}$ be i.i.d.\ draws from the target $\pi$. Define the associated empirical measures
\[
\widehat\mu_N := \frac{1}{N}\sum_{k=1}^N \delta_{X_k},
\qquad
\widehat\nu_N := \frac{1}{N}\sum_{k=1}^N \delta_{Y_k}.
\]
In one dimension, the Wasserstein-$1$ distance between empirical measures with equal sample size admits the explicit expression
\[
W_1(\widehat\mu_N,\widehat\nu_N)=\frac{1}{N}\sum_{i=1}^N |X_{(i)}-Y_{(i)}|.
\]
 We use step sizes $h=2^{-k}$ for $13$ equally spaced values of $k \in [1,5]$, and for each $h$ draw $N=2 \times 10^8$ samples after burn-in. We report the resulting Wasserstein-$1$ asymptotic bias for three stochastic-gradient methods: stochastic gradient Langevin dynamics (SGLD); a semi-implicit Euler discretisation of kinetic Langevin dynamics in the spirit of SG-HMC \citep{SGHMC}, labelled SG-HMC in Figure~\ref{fig:bias-three-gammas},
\[
v_{k+1}=(1-h\gamma)v_k-h\,\mathcal G(x_k,\omega_{k+1})+\sqrt{2\gamma h}\,\xi_{k+1},
\qquad
x_{k+1}=x_k+h\,v_{k+1};
\]
and the stochastic-gradient UBU integrator (SG-UBU). At the largest stepsizes some of the schemes diverge numerically; such points, and points whose bias exceeds the plotted range, are not shown.

\begin{figure}[t]
    \centering
    \includegraphics[width=\textwidth]{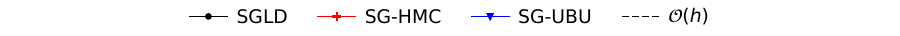}\\[2pt]
    \begin{subfigure}[t]{0.32\textwidth}
        \centering
        \includegraphics[width=\textwidth]{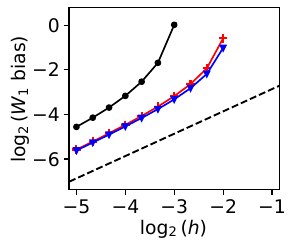}
        \caption{$\gamma=1.0$}
        \label{fig:bias-gamma-1p0}
    \end{subfigure}\hfill
    \begin{subfigure}[t]{0.32\textwidth}
        \centering
        \includegraphics[width=\textwidth]{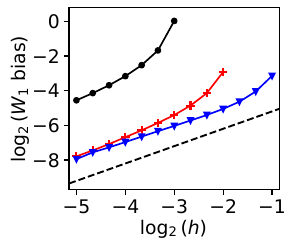}
        \caption{$\gamma=5$}
        \label{fig:bias-gamma-5}
    \end{subfigure}\hfill
    \begin{subfigure}[t]{0.32\textwidth}
        \centering
        \includegraphics[width=\textwidth]{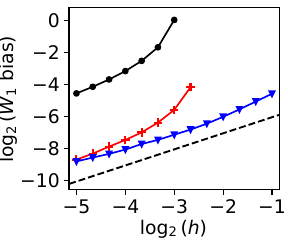}
        \caption{$\gamma=10$}
        \label{fig:bias-gamma-10}
    \end{subfigure}

    \caption{Empirical Wasserstein-$1$ asymptotic bias versus stepsize for SGLD, SG-HMC (semi-implicit Euler) and SG-UBU on the one-dimensional Gaussian target described in Section~\ref{sec:numevalbias}, for three values of the friction parameter $\gamma$.}
    \label{fig:bias-three-gammas}
\end{figure}

\subsection{Bayesian logistic regression on MNIST 3-vs-5}
We consider Bayesian logistic regression for binary classification of the digits $3$ and $5$ from the MNIST training set. The dataset contains $N=11552$ observations and $d = 784$. Writing the parameter as $q \in \mathbb{R}^d$, we use the Gaussian prior
\[
q \sim \mathcal{N}(0,\sigma^2 I_d), \qquad \sigma^2 = 10^{-3},
\]
so that the negative log-posterior, $V:\mathbb{R}^{d} \to \mathbb{R}$ is defined pointwise by
\[
V(q)
=
\frac{1}{2\sigma^2}\|q\|^2
+
\sum_{i=1}^N
\left[
\log(1+\exp(x_i^\top q)) - y_i x_i^\top q
\right].
\]
To construct the control-variate stochastic gradient, we first compute an approximation $q_{\min}$ of the minimiser of $V$ using BFGS. Then, letting
\[
\nabla \ell_i(q)=x_i\left(\sigma(x_i^\top q)-y_i\right),
\qquad
\sigma(t)=\frac{1}{1+e^{-t}},
\]
the control variate stochastic gradient estimator based on a minibatch $B_k$ of size $b$ (sampled uniformly with replacement; we use $b=100$) is
\[
\widehat{\nabla V}(q)
=
\frac{1}{\sigma^2}q
+
\nabla \ell(q_{\min})
+
\frac{N}{b}\sum_{i\in B_k}\left(\nabla \ell_i(q)-\nabla \ell_i(q_{\min})\right),
\]
where $\nabla \ell(q_{\min})=\sum_{i=1}^N \nabla \ell_i(q_{\min})$.

For this Bayesian logistic-regression model, the Gaussian prior makes the negative log-posterior globally strongly convex, while the logistic likelihood has a globally bounded Hessian. The minibatch control-variate estimator is unbiased, continuously differentiable, and has finite gradient-noise moments and bounded Jacobian variance. Thus the structural assumptions of Theorem~\ref{theorem:bias_sgUBU} on $V$ and the stochastic gradient hold.

We compare SG-UBU with the Euler--Maruyama discretisation of \eqref{eq:kinetic_langevin} (SG-EM, labelled SG-HMC in Figure~\ref{fig:bias-BLR}), $x_{k+1}=x_k+hv_k$, $v_{k+1}=(1-h\gamma)v_k-h\,\widehat{\nabla V}(x_k)+\sqrt{2\gamma h}\,\xi_{k+1}$, both targeting $\pi(q)\propto \exp{\left(-V(q)\right)}$. The friction parameter is chosen as $\gamma=\sqrt{L_{\mathrm{loc}}}$, where $L_{\mathrm{loc}}$ is the largest eigenvalue of the Hessian $\nabla^2 V(q_{\min})$ at the minimiser, a local curvature estimate rather than the global constant $L$ of Assumption~\ref{ass:strong-log-conc-log-smooth}. We consider the step sizes
$
h \in \left\{
\frac{2}{\sqrt{L_{\mathrm{loc}}}},
\frac{1}{\sqrt{L_{\mathrm{loc}}}},
\frac{1}{2\sqrt{L_{\mathrm{loc}}}},
\frac{1}{4\sqrt{L_{\mathrm{loc}}}},
\frac{1}{8\sqrt{L_{\mathrm{loc}}}},
\frac{1}{16\sqrt{L_{\mathrm{loc}}}}
\right\}.
$
This tuning lies outside the regime of Theorem~\ref{theorem:bias_sgUBU}, which requires $\gamma\geq\sqrt{8L}$ and $h<1/(2\gamma)$. For SG-EM, $h=2/\sqrt{L_{\mathrm{loc}}}$ gives $h\gamma=2$, the stability limit of the explicit velocity update, so its bias at that stepsize reflects instability rather than discretisation error.

For each method and each value of $h$, we run $100$ independent chains started at $q_{\min}$ with zero velocity, each for $10^6$ steps, discard a burn-in of $10^4$ steps, and estimate $\mathbb{E}_{\pi}[V(q)]$ by averaging $V(q_k)$ over every tenth iterate. As a reference value, we use an SG-UBU simulation with the same control-variate gradient estimator at the smallest step size $h_{\mathrm{ref}}=1/(16\sqrt{L_{\mathrm{loc}}})$, with $100$ chains of $2\times10^6$ steps and a burn-in of $2\times10^4$ steps. For both methods, the reported error is
$
\left|
\widehat{\mathbb{E}}_{h}[V]-\widehat{\mathbb{E}}_{\mathrm{ref}}[V]
\right|,
$
which we plot against $h$ on log-log scales at the five step sizes larger than $h_{\mathrm{ref}}$; the point at $h=h_{\mathrm{ref}}$ is omitted, since its bias estimate against the same-stepsize reference is dominated by Monte Carlo noise. Error bars are standard errors across the independent chains, including the uncertainty of the reference run.

This quantity estimates $\mathbb{E}_{\pi_h}[V]-\mathbb{E}_{\pi_{h_{\mathrm{ref}}}}[V]$ rather than the bias against $\pi$. For SG-UBU, the measured error decreases approximately like $h^2$ for $h\sqrt{L_{\mathrm{loc}}}\in[1/2,2]$, and at the two smallest plotted stepsizes it is not distinguishable from zero within the error bars. Extrapolating this second-order decay, the bias of the reference itself is about $4\times10^{-3}$, an order of magnitude below its Monte Carlo standard error ($\approx0.04$); even if the bias were only first order, the measurement at $h=2h_{\mathrm{ref}}$ would bound it by about $0.15$ (two standard errors), well below the errors at the three largest stepsizes. The SG-EM errors decrease approximately linearly in $h$ at the smaller stepsizes. Finally, $V$ is not globally Lipschitz, so this observable is not directly covered by the $W_1$ guarantee; the experiment illustrates the order of the bias rather than checking the bound.
\begin{figure}[t]
    \centering
    \includegraphics[width=0.4\textwidth]{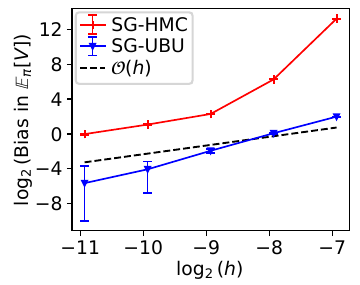}
    \caption{Empirical error in the expected negative log-density versus stepsize for SG-EM (labelled SG-HMC) and SG-UBU with $\gamma=\sqrt{L_{\mathrm{loc}}}$, measured against an SG-UBU reference run at $h_{\mathrm{ref}}=1/(16\sqrt{L_{\mathrm{loc}}})$. Error bars are standard errors across $100$ independent parallel chains, including the uncertainty of the reference run.}
    \label{fig:bias-BLR}
\end{figure}

\subsection{Random-effects logistic regression on the Indonesian xerophthalmia data}\label{sec:xerop}

Our final example is the random-effects logistic regression model for childhood respiratory infection from the Indonesian xerophthalmia study \citep{sommer1983increased}, as analysed in \cite{diggle2002analysis} and used as a standard benchmark for pseudo-marginal samplers in  \cite{schmon2021large} and \cite{alenlov2021pseudo}. The data contain $N_{\mathrm{obs}}=1200$ binary observations of respiratory disease on $\ell=275$ children with up to $J_{\max}=6$ visits each, together with $p=8$ child-level covariates (intercept plus age, xerophthalmia indicator, two seasonal harmonics, sex, height-for-age and stunting). Writing $y_{t,j}\in\{0,1\}$ for the $j$-th observation on child $t$ and $x_{t,j}\in\R^{p}$ for the corresponding covariate vector, the model is
\begin{equation}\label{eq:xerop-model}
y_{t,j}\mid \beta,b_t \sim \mathrm{Bernoulli}\left(\sigma(x_{t,j}^{\top}\beta+b_t)\right),\qquad
b_t\mid \tau \stackrel{\mathrm{iid}}{\sim}\mathcal{N}(0,\tau),\qquad t=1,\dots,\ell,
\end{equation}
where $\sigma(s)=(1+e^{-s})^{-1}$. We place a weak Gaussian prior $\beta\sim\mathcal{N}(0,10^{4}I_{p})$ and an inverse-gamma prior $\tau\sim\mathrm{InvGamma}(1,1.5)$ truncated to $\tau\geq 0.05$, and parametrise via $\xi=\log(\tau-0.05)$ so that the unknown is $\theta=(\beta,\xi)\in\R^{d}$ with $d=9$. The marginal log-posterior is
\begin{equation}\label{eq:xerop-marginal}
\log\pi(\theta)
=
\log p(\theta)
+\sum_{t=1}^{\ell}\log\int_{\R}\Bigl[\prod_{j=1}^{J_{t}}\sigma\left(x_{t,j}^{\top}\beta+b\right)^{y_{t,j}}\left(1-\sigma\left(x_{t,j}^{\top}\beta+b\right)\right)^{1-y_{t,j}}\Bigr]\mathcal{N}(b\mid 0,\tau)db,
\end{equation}
which is intractable in closed form. The pseudo-marginal samplers of  \cite{schmon2021large} and \cite{alenlov2021pseudo} target the extended density $\widetilde{\pi}(\theta,u)\propto p(\theta)\widehat{L}(\theta;u)\,m_{\theta}(u)$ on $\R^{d}\times\R^{\ell N}$, where $\widehat{L}$ is an unbiased $N$-particle importance-sampling estimator of the marginal likelihood appearing in \eqref{eq:xerop-marginal}, $u\in\R^{\ell N}$ are the auxiliary draws and $m_{\theta}$ is their generating density; unbiasedness ensures that the $\theta$-marginal of $\widetilde{\pi}$ is the posterior $\pi(\theta)$. With $V(\theta)=-\log\pi(\theta)$, an SG-UBU method based on an exact unbiased stochastic gradient of $V$ would avoid this augmentation and evolve directly on the $\theta$-space.

\paragraph{Stochastic gradient.}
For each child $t$, the conditional density of the random effect is $p(b\mid \theta,y)\propto\exp(f_{t}(b;\theta))$ with
\[
f_{t}(b;\theta)=\Big(\sum_{j}y_{t,j}\Big)b-\sum_{j}\log\left(1+\exp(x_{t,j}^{\top}\beta+b)\right)-\frac{b^{2}}{2\tau},
\]
which is one-dimensional and strictly log-concave. By Fisher's identity, $\nabla V(\theta)$ is an expectation over these conditionals. We estimate it by importance sampling from a two-tangent piecewise-exponential proposal, using $n_{\mathrm{re}}=6$ stratified antithetic draws per child and a $K=8$ point Gauss--Hermite quadrature for the normalising constants; see Appendix~\ref{app:xerop} for the precise estimator \eqref{eq:xerop-grad}. Because of the quadrature, the estimator is only approximately unbiased, although the relative quadrature error is of order $10^{-4}$ (below $6\times10^{-4}$ at the reported posterior means). We have not established Assumption~\ref{Assumption:Bounded_Variance} or global strong convexity for this example, so none of our theoretical bounds is invoked here.

\paragraph{Tuning and reference values.}
We use friction $\gamma=2$ and precondition SG-UBU with the mass matrix $M=\widehat{\Sigma}^{-1}$, where $\widehat{\Sigma}$ is the posterior covariance estimated during an unadjusted warm-up; the $7500$-step burn-in is split equally into an initial phase using the mode curvature $-\nabla^{2}\log\pi(\theta_{\mathrm{MAP}})$ (computed via Louis' identity \citep[see][]{louis1982finding}), a covariance-estimation window, and a further burn-in in the adapted metric. The stepsize $h=0.45$ was chosen empirically by stability and efficiency checks, since Theorem~\ref{theorem:bias_sgUBU} has not been verified for this posterior. As a reference we use a converged run of pseudo-marginal HMC with a randomised number of leapfrog steps $L\sim\mathrm{Geometric}(1/\bar{L})$, $\bar{L}=12$, $\varepsilon=0.8$ and $N=30$ particles (PM-RHMC). We also run the random-walk pseudo-marginal Metropolis--Hastings algorithm (PMMH) of \cite{schmon2021large} with $N=30$ and their optimally scaled proposal (step scale $2.4/\sqrt{d}$).

\paragraph{Numerical results.}
Table~\ref{tab:xerop-summary} reports posterior means and standard deviations for the three samplers, and Table~\ref{tab:xerop-efficiency} summarises their efficiency. SG-UBU delivers an effective sample size per second more than an order of magnitude larger than PMMH and nearly $40$ times that of PM-RHMC, while requiring fewer than $10$ stochastic-gradient evaluations per effective sample.

\begin{table}[ht]
\centering\small
\setlength{\tabcolsep}{4pt}
\resizebox{\linewidth}{!}{%
\begin{tabular}{l rr rr rr rr}
\toprule
& \multicolumn{2}{c}{PMMH ($N=30$)}
& \multicolumn{2}{c}{PM-RHMC ($\bar{L}=12$)}
& \multicolumn{2}{c}{SG-UBU ($h=0.45$)}
& \multicolumn{2}{c}{reference \citep{schmon2021large}}\\
\cmidrule(lr){2-3}\cmidrule(lr){4-5}\cmidrule(lr){6-7}\cmidrule(lr){8-9}
parameter & mean & sd & mean & sd & mean & sd & mean & sd\\
\midrule
$\beta_{0}$              & $-2.788$ & $0.230$ & $-2.789$ & $0.231$ & $-2.789$ & $0.228$ & $-2.788$ & $0.230$ \\
$\beta_{\text{age}}$     & $-0.035$ & $0.008$ & $-0.035$ & $0.008$ & $-0.035$ & $0.008$ & $-0.035$ & $0.008$ \\
$\beta_{\text{xero}}$    & $+0.559$ & $0.506$ & $+0.556$ & $0.507$ & $+0.558$ & $0.498$ & $+0.560$ & $0.507$ \\
$\beta_{\cos}$           & $-0.615$ & $0.179$ & $-0.615$ & $0.178$ & $-0.615$ & $0.176$ & $-0.614$ & $0.178$ \\
$\beta_{\sin}$           & $-0.174$ & $0.179$ & $-0.172$ & $0.178$ & $-0.173$ & $0.176$ & $-0.173$ & $0.179$ \\
$\beta_{\text{sex}}$     & $-0.461$ & $0.275$ & $-0.462$ & $0.274$ & $-0.460$ & $0.271$ & $-0.461$ & $0.276$ \\
$\beta_{\text{height}}$  & $-0.052$ & $0.028$ & $-0.052$ & $0.028$ & $-0.052$ & $0.028$ & $-0.052$ & $0.028$ \\
$\beta_{\text{stunted}}$ & $+0.192$ & $0.464$ & $+0.196$ & $0.464$ & $+0.192$ & $0.458$ & $+0.192$ & $0.466$ \\
$\tau$                   & $+0.943$ & $0.366$ & $+0.944$ & $0.369$ & $+0.943$ & $0.371$ & $+0.944$ & $0.367$ \\
\bottomrule
\end{tabular}
}
\caption{Posterior means and standard deviations for the xerophthalmia random-effects model \eqref{eq:xerop-model}. The reference column gives the long-run pseudo-marginal benchmark values of \citet{schmon2021large}. Across the nine parameters, the maximum error in the posterior means, relative to the reference posterior standard deviation, is $0.4\%$ (PMMH), $1.0\%$ (PM-RHMC) and $0.7\%$ (SG-UBU); for the posterior standard deviations it is $0.5\%$, $0.9\%$ and $1.9\%$.}
\label{tab:xerop-summary}
\end{table}

\begin{table}[ht]
\centering\small
\setlength{\tabcolsep}{5pt}
\begin{tabular}{l c r r c c r r}
\toprule
sampler & chains & burnin/chain & samples/chain & accept & $\widehat{R}_{\max}$ & min ESS & min ESS/s\\
\midrule
PMMH        & $14$ & $10000$ & $1000000$ & $0.148$ & $1.000$ & $187451$   & $297$ \\
PM-RHMC & $14$ & $1000$  & $10000$     & $0.791$ & $1.000$ & $37730$    & $95$ \\
SG-UBU      & $14$ & $7500$  & $1000000$ & ---     & $1.000$ & $1483082$ & $3711$ \\
\bottomrule
\end{tabular}
\caption{Sampling efficiency for the xerophthalmia example. Burn-in and samples per chain are integrator steps for SG-UBU and Metropolis sweeps for PMMH and PM-RHMC. Minimum ESS and ESS/s are over the nine coordinates; all runs were timed on the same machine (M4 Pro, 14-core CPU) after JIT compilation.}
\label{tab:xerop-efficiency}
\end{table}

The pseudo-marginal samplers are penalised by the latent vector $u\in\R^{\ell N}=\R^{8250}$ in their extended target: PMMH needs a fresh $N$-particle likelihood estimate at every step, and the variance of $\log\widehat{L}$ lowers its acceptance rate \citep{schmon2021large}, while PM-RHMC must integrate Hamiltonian dynamics in the $\ell N$-dimensional auxiliary space. SG-UBU avoids the augmentation: the latent variables enter only through the stochastic-gradient estimator, whose cost per child does not depend on any particle count $N$.

\section{Conclusion}\label{sec:conclusion}

We established $\mathcal{O}(h)$ Wasserstein bounds on the asymptotic bias of stochastic-gradient UBU under mild assumptions on the stochastic-gradient noise. The main technical ingredient is a family of Gaussian convolution inequalities, which allow the stochastic-gradient contribution to be controlled directly at stationarity. Combined with existing contraction estimates, these results yield explicit non-asymptotic guarantees for stochastic-gradient kinetic Langevin Monte Carlo. The Bayesian logistic-regression example satisfies the structural assumptions, whereas the random-effects example is presented only as an empirical study outside the presently verified theory. While we have worked in the strongly log-concave setting in order to obtain quantitative bounds, it is also possible to extend the analysis to non-convex regimes using results in \cite{schuh2024convergence}. For fixed problem parameters, suppose that the stochastic-gradient assumptions used in Theorem~\ref{theorem:bias_sgUBU} give a stationary bias bounded by $C h$ and that the initial Wasserstein distance is bounded uniformly for sufficiently small $h$. Proposition~\ref{prop:wass_SGUBU} then yields, for constants $c,C,C'>0$,
\[
W_{1,a,b}(\nu P_h^n,\overline\pi)
\leq C' e^{-cnh}+Ch.
\]
Choosing $h=\mathcal O(\epsilon)$ and then $n=\mathcal O\!\left(h^{-1}\log(1/\epsilon)\right)$ makes both terms at most of order $\epsilon$. Since SG-UBU uses one stochastic-gradient evaluation per step, its resulting complexity is therefore
\[
n=\mathcal O\!\left(\epsilon^{-1}\log\frac1\epsilon\right).
\]
Further, by optimising our bounds SG-UBU in the setting where $C_{G} \lesssim L^{2}$ one can reach a sufficiently small accuracy $\epsilon > 0$ in $\sqrt{m}W_1$ in
\[
n=\widetilde{\mathcal{O}}\left(\frac{\kappa^{3/2} d^{1/2}}{\varepsilon}+\frac{\kappa^{1/2}\sigma^2_{2}}{m \varepsilon}\right),
\]
iterations up to logarithmic factors. The $\epsilon^{-1}$ polynomial dependence is consistent with the recent complexity lower bounds of \cite{chen2026high}, as discussed in \cite{che-book-2024+}. Appendix~\ref{sec:gaussian-sharpness} shows in an exactly solvable Gaussian example that the first-order stationary-bias dependence used in this choice of $h$ is sharp up to constants.

The framework is also broad enough to cover stochastic gradients arising from subsampling procedures beyond standard minibatching in Bayesian inference, including the random batch method from computational statistical physics \citep{jin2020random}. Recent works \citep[see][]{paulin2025sampling,shaw2025random} suggest that epoch-wise without-replacement subsampling strategies may improve the order of the asymptotic bias. However, a gap remains between the strong and weak orders of accuracy in this setting. A natural direction for future work is therefore to investigate whether one can prove $\mathcal{O}(h^2)$ Wasserstein bias bounds for stochastic-gradient UBU under without-replacement subsampling. Such second-order behaviour has also been observed empirically in \cite{paulin2025sampling}.

\section*{Acknowledgements and Disclosure of Funding}
DP was supported by a Nanyang Technological University Start-up Grant, project number: 024968-00001. The authors declare no competing interests. {We would like to thank Yuansi Chen for the stochastic localization approach for proving the Gaussian convolution inequalities. We have used ChatGPT and Claude for coding assistance, proof checking and for improving constants in bounds.}

\begin{appendices}

\section{Non-asymptotic guarantees for stochastic gradient UBU}\label{sec:nonasymp}
Recall the representation \eqref{eq:disc} of UBU and SG-UBU, with $\mathcal E,\mathcal F$ as in \eqref{eq:EFdef}. To control the one-step error from the target measure, we use that under $\overline\pi$ the velocity marginal is Gaussian, and the It\^o integrals contribute additional Gaussian increments. The next proposition combines these contributions as a single block-Gaussian random vector.
\begin{proposition}\label{prop:block-gaussian}
Let $W=(W_t)_{t\geq 0}$ be a standard $d$-dimensional Brownian motion and let $v\sim\mathcal N(0,I_d)$ be independent of $W$. Fix $\gamma>0$ and $h>0$, and define $\mathcal E,\mathcal F$ as in \eqref{eq:EFdef}. Define the $\mathbb{R}^{2d}$-valued random vector
\[
X
:=
\binom{X_1}{X_2}
:=
\binom{
\mathcal F(h)v+\sqrt{2\gamma}\int_0^h
\mathcal F(h-s)\,dW_s
}{
\mathcal E(h)v+\sqrt{2\gamma}\int_0^h
\mathcal E(h-s)\,dW_s
}.
\]
Then $X$ is a centred Gaussian and
\[
X \ \stackrel{d}{=} \ \Sigma(h,\gamma)Z_{2d},
\]
where $Z_{2d}\sim\mathcal N(0,I_{2d})$ and
\[
\Sigma(h,\gamma)
=
\begin{pmatrix}
\mathcal F(h)I_d & \sigma(h,\gamma)I_d\\[2mm]
I_d & 0
\end{pmatrix},
\]
with
\begin{equation}\label{eq:sigma-def}
\sigma(h,\gamma)^2 := 2\gamma \int_0^h \mathcal F(u)^2du
= \frac{2h}{\gamma}-\frac{3}{\gamma^2}+\frac{4e^{-\gamma h}}{\gamma^2}-\frac{e^{-2\gamma h}}{\gamma^2}.
\end{equation}
Equivalently,
\[
\operatorname{Cov}(X)
=
\begin{pmatrix}
\left(
\frac{2h}{\gamma}
-\frac{2(1-e^{-\gamma h})}{\gamma^2}
\right)I_d
&
\mathcal F(h)I_d\\[2mm]
\mathcal F(h)I_d&I_d
\end{pmatrix}.
\]
\end{proposition}

\begin{proof}
$X$ is a linear image of the independent Gaussian objects $v$ and $W$, hence centred Gaussian. By independence and It\^o's isometry, $\mathrm{Cov}(\int_0^h k\,dW,\int_0^h\ell\,dW)=(\int_0^h k\ell\,ds)I_d$ for deterministic $k,\ell\in L^2([0,h])$, so
\begin{align*}
\mathrm{Cov}(X_2)&=\Big(\mathcal E(h)^2+2\gamma\int_0^h\mathcal E(u)^2du\Big)I_d=I_d,\\
\mathrm{Cov}(X_1,X_2)&=\Big(\mathcal F(h)\mathcal E(h)+2\gamma\int_0^h\mathcal E(u)\mathcal F(u)du\Big)I_d=\mathcal F(h)I_d,\\
\mathrm{Cov}(X_1)&=\Big(\mathcal F(h)^2+2\gamma\int_0^h\mathcal F(u)^2du\Big)I_d=\big(\mathcal F(h)^2+\sigma(h,\gamma)^2\big)I_d,
\end{align*}
where the integrals are elementary; the last expression equals the stated $(1,1)$ block. Since $\Sigma(h,\gamma)\Sigma(h,\gamma)^\top$ has exactly these blocks, $X\stackrel{d}{=}\Sigma(h,\gamma)Z_{2d}$.
\end{proof}

Using this Proposition we now apply the Gaussian convolution inequalities in the local error analysis as follows.

\subsection{Proof of Theorem~\ref{theorem:bias_sgUBU}}

\sgubuBias*

\begin{proof}[Proof of Theorem~\ref{theorem:bias_sgUBU}]
{The proof has three steps. First, we decompose the invariant-measure error into a deterministic UBU local-bias term, a one-step stochastic-gradient perturbation term, and a contraction term. Second, we use the contraction estimate for SG-UBU to absorb the final term. Third, we bound the one-step stochastic-gradient perturbation by coupling the Gaussian noise optimally with the convolution induced by the centred stochastic-gradient error.}

Let $P$ denote the transition kernel of UBU and $P_{\mathrm{SG}}$ the
transition kernel of SG-UBU, both given by \eqref{eq:disc}, then
\begin{equation}\label{eq:decomp_Wp}
    W_{p,a,b}(\overline{\pi},\overline{\pi}_{h})
    \leq
    W_{p,a,b}(\overline{\pi},\overline{\pi} P)
    +
    W_{p,a,b}(\overline{\pi} P,\overline{\pi} P_{\mathrm{SG}})
    +
    W_{p,a,b}(\overline{\pi} P_{\mathrm{SG}},\overline{\pi}_{h}P_{\mathrm{SG}}).
\end{equation}
By \cite[Proposition I.3]{ububu}, for $h < \min\{\frac1{2\gamma},\frac1{5\sqrt L}\}$ (the second condition follows from the first, since $2\gamma\geq 2\sqrt{8L}>5\sqrt L$),
\[
W_{p,a,b}(\overline{\pi},\overline{\pi} P)
\leq \frac37 \sqrt d(\sqrt L+\gamma)h^2.
\]
By Proposition~\ref{prop:wass_SGUBU},
\[
W_{p,a,b}(\overline{\pi} P_{\mathrm{SG}},\overline{\pi}_{h}P_{\mathrm{SG}})
\le
\left(1-\frac{mh}{4\gamma}+\frac{8h^2C_G}{7L}\right)^{1/2}
W_{p,a,b}(\overline{\pi},\overline{\pi}_h).
\]
Hence
\begin{align*}
W_{p,a,b}(\overline{\pi},\overline{\pi}_h)
\le
\frac{
\frac37 \sqrt d(\sqrt L+\gamma)h^2 + W_{p,a,b}(\overline{\pi}P,\overline{\pi}P_{\mathrm{SG}})
}{
1-\left(1-\frac{mh}{4\gamma}+\frac{8h^2C_G}{7L}\right)^{1/2}
}.
\end{align*}
Set
\[
r_h:=\frac{mh}{4\gamma}-\frac{8h^2C_G}{7L}
=\frac{h\left(mL-\frac{32}{7}hC_G\gamma\right)}{4\gamma L}.
\]
Under the assumption
$
h<\min\left\{\frac{1}{2\gamma},\frac{7mL}{32C_G\gamma}\right\},
$
we have $r_h>0$. Moreover, since $m\leq L$, $h<1/(2\gamma)$, and $\gamma^2\geq 8L$,
 $0<r_h<1$, and so
\[
1-\sqrt{1-r_h}\geq \frac{r_h}{2}
=
\frac{h\left(mL-\frac{32}{7}hC_G\gamma\right)}{8\gamma L}.
\]
Consequently,
\begin{equation}\label{eq:denominator_bound}
\frac{1}{1-\left(1-\frac{mh}{4\gamma}+\frac{8h^2C_G}{7L}\right)^{1/2}}
\le
\frac{8\gamma L}{h\left(mL-\frac{32}{7}hC_G\gamma\right)}.
\end{equation}
Thus
\begin{equation}\label{eq:prelocal}
W_{p,a,b}(\overline{\pi},\overline{\pi}_h)
\le
\frac{8\gamma L}{h\left(mL-\frac{32}{7}hC_G\gamma\right)}
\left[
\frac37 \sqrt d(\sqrt L+\gamma)h^2
+
W_{p,a,b}(\overline{\pi}P,\overline{\pi}P_{\mathrm{SG}})
\right].
\end{equation}

It remains to bound $W_{p,a,b}(\overline{\pi}P,\overline{\pi}P_{\mathrm{SG}})$.
For each $x\in\R^d$, let $\eta_x$ and $\eta_x^{SG}$ denote the conditional one-step laws of UBU
and SG-UBU started from position $x$, with initial velocity distributed according to
$\g:=\mathcal N(0,I_d)$. Then
\[
\overline{\pi}P=\int \eta_x\pi(dx),
\qquad
\overline{\pi}P_{\mathrm{SG}}=\int \eta_x^{SG}\pi(dx),
\]
and
\cite[Theorem~4.8]{villani2009optimal} yields
\[
W_{p,a,b}^p(\overline{\pi}P,\overline{\pi}P_{\mathrm{SG}})
\le
\int_{\R^d} W_{p,a,b}^p(\eta_x,\eta_x^{SG})\pi(dx).
\]

For $x\in\R^d$ we now wish to bound $W_{p,a,b}(\eta_x,\eta_x^{SG})$. Let $v\sim\mathcal N(0,I_d)$, and $(W_t)_t$ be a standard $d$-dimensional Brownian
motion independent of $v$, and define
\begin{align*}
Y_x &:=
x+\mathcal F(h/2)v+\sqrt{2\gamma}\int_0^{h/2}\mathcal F(h/2-s)dW_s,\\
N_x
&:=
\binom{N_x^{(1)}}{N_x^{(2)}}
:=
\binom{
\mathcal F(h)v+\sqrt{2\gamma}\int_0^h
\mathcal F(h-s)\,dW_s
}{
\mathcal E(h)v+\sqrt{2\gamma}\int_0^h
\mathcal E(h-s)\,dW_s
}.
\end{align*}
By \eqref{eq:disc}, $\eta_x$ is the law of $\big(x+N_x^{(1)}-h\mathcal F(h/2)\nabla V(Y_x),\,N_x^{(2)}-h\mathcal E(h/2)\nabla V(Y_x)\big)$, and $\eta_x^{SG}$ is the law of the same vector with $\nabla V(Y_x)$ replaced by $\mathcal G(Y_x,\omega)$ for an independent $\omega\sim\Gamma$. By Proposition~\ref{prop:block-gaussian}, $N_x\stackrel d=\Sigma Z_{2d}$, where $\Sigma:=\Sigma(h,\gamma)$, $Z_{2d}=(Z_d,\widehat Z_d)\sim\mathcal N(0,I_{2d})$ and $\Sigma Z_{2d}=(\mathcal F(h)Z_d+\sigma(h,\gamma)\widehat Z_d,\,Z_d)$. Hence, whenever
\[
(y,\Sigma(Z_d,\widehat Z_d))\stackrel d=(Y_x,N_x),
\qquad
(\widetilde y,\Sigma(Z'_d,\widehat Z'_d),\omega)\stackrel d=(Y_x,N_x,\omega),
\]
$\eta_x$ and $\eta_x^{SG}$ are, respectively, the laws of
\begin{align*}
&\big(x+\mathcal F(h)Z_d+\sigma(h,\gamma)\widehat Z_d-h\mathcal F(h/2)\nabla V(y),\ Z_d-h\mathcal E(h/2)\nabla V(y)\big),\\
&\big(x+\mathcal F(h)Z'_d+\sigma(h,\gamma)\widehat Z'_d-h\mathcal F(h/2)\mathcal G(\widetilde y,\omega),\ Z'_d-h\mathcal E(h/2)\mathcal G(\widetilde y,\omega)\big).
\end{align*}

Now define
\[
A_x:=h\mathcal E(h/2)\big(\nabla V(x)-\mathcal G(x,\omega)\big),
\qquad
\mu_x:=\mathrm{Law}(A_x),
\]
and
\[
\sigma_p(x):=\Big(\E_\omega\|\nabla V(x)-\mathcal G(x,\omega)\|^p\Big)^{1/p}.
\]
Set $t:=h/2$ and $\eta:=\mathcal E(t)=e^{-\gamma h/2}$, then we first make the conditional coupling of the intermediate positions explicit. Let $U_x:=Y_x-x$, then since $(U_x,Z_d,\widehat Z_d)$ is jointly centred Gaussian and its covariance is block-isotropic, one can find scalars $c_1,c_2$ and a centred Gaussian random vector $R_d$, independent of $(Z_d,\widehat Z_d)$, such that
\[
U_x=c_1Z_d+c_2\widehat Z_d+R_d.
\]
Moreover, using that the stationary OU velocity has covariance $\operatorname{Cov}(V_s,V_h)=e^{-\gamma(h-s)}I_d$ for $0\leq s\leq t$,
\[
c_1 I_d
=
\operatorname{Cov}(U_x,Z_d)
=
\operatorname{Cov}(Y_x-x,N_x^{(2)})
=
\left(\int_0^t e^{-\gamma(h-s)}\,ds\right)I_d
=
\eta\mathcal F(t)I_d,
\]
and the value of $c_2$ will not be needed.

For this fixed $x$, let $\lambda_x\in\Pi(\mu_x*\g,\g)$ be an optimal coupling. Realising its first marginal as $A_x(\omega)+Z'_d$, where $(\omega,Z'_d)\sim\Gamma\otimes\g$, choose $Z_d$ such that
\[
(A_x(\omega)+Z'_d,Z_d)\sim\lambda_x.
\]
Independently of $(\omega,Z'_d,Z_d)$, sample $\widehat Z_d\sim\g$ and a Gaussian $R_d$ having the law above, and define
\begin{align*}
 y&:=x+c_1Z_d+c_2\widehat Z_d+R_d,\\
 \widetilde y&:=x+c_1Z'_d+c_2\widehat Z_d+R_d.
\end{align*}
Then
\[
 (y,\Sigma(Z_d,\widehat Z_d))\stackrel d=(Y_x,N_x),
 \qquad
 (\widetilde y,\Sigma(Z'_d,\widehat Z_d),\omega)\stackrel d=(Y_x,N_x,\omega).
\]
Note that $Z'_d$, $\widehat Z_d$, and $R_d$ are independent of $\omega$, so $\widetilde y$ remains independent of $\omega$. This coupling yields
\begin{equation}\label{eq:refined_y_coupling}
 y-\widetilde y
 =\eta\mathcal F(h/2)(Z_d-Z'_d).
\end{equation}

Now define
\[
 D_x:=A_x+Z'_d-Z_d,
 \qquad
 \mathcal R_x(\omega):=\nabla V(x)-\mathcal G(x,\omega), 
\]
and 
\[
 \Delta_x:=\nabla V(y)-\mathcal G(\widetilde y,\omega)
 -\left(\nabla V(x)-\mathcal G(x,\omega)\right).
\]
Using $A_x=h\eta\mathcal R_x(\omega)$ and $\mathcal F(h)=(1+\eta)\mathcal F(h/2)$, the coupled one-step difference can be written as
\begin{equation}\label{eq:local_error_refined}
\binom{\widetilde x_1-x_1}{\widetilde v_1-v_1}
=
\binom{\mathcal F(h)D_x}{D_x}
+
\binom{h\mathcal F(h/2)\Delta_x}{h\eta\Delta_x}
+
\binom{
h\mathcal F(h/2)(1-\eta-\eta^2)\mathcal R_x(\omega)
}{0}.
\end{equation}
This is a coupling of $\eta_x$ and $\eta_x^{SG}$ by the two marginal identities above.

We now estimate the three vectors in \eqref{eq:local_error_refined} directly in the twisted norm, $\|(x,v)\|_{a,b}$. Firstly, we have
\begin{align*}
W_{p,a,b}(\eta_x,\eta_x^{SG}) &\le q_0\,\|D_x\|_{L^p} +hq_1\,\|\Delta_x\|_{L^p}\\
&\quad + h\mathcal F(h/2)|1-\eta-\eta^2|\,\sigma_p(x),
\end{align*}
where
\[
q_0^2=L^{-1}+2\gamma^{-1}\mathcal F(h)+\mathcal F(h)^2,
\qquad
q_1^2=\eta^2L^{-1}+2\eta\gamma^{-1}\mathcal F(h/2)+\mathcal F(h/2)^2.
\]
By optimality of $\lambda_x$,
\[
\|D_x\|_{L^p}=W_p(\mu_x*\g,\g).
\]
Let $r:=L/\gamma^2$, and $q:=1-e^{-\gamma h}$, then since $\gamma h\le1/2$, $r\le1/8$, and $q\le2/5$,
\[
Lq_0^2=1+r(2q+q^2)\leq 1+\frac{3}{25}<\left(\frac{17}{16}\right)^2, \qquad q_0\le\frac{17}{16\sqrt L}.
\]
Also, writing $f:=1-\eta$ and using $\mathcal F(h/2)=f/\gamma$,
\[
Lq_1^2=\eta^2+r(1-\eta^2)\leq 1, \qquad q_1\leq\frac1{\sqrt L}.
\]
Finally, the assumptions imply $Lh^2\leq 1/32$, and we shall also use $\mathcal F(h/2)\leq h/2$.

Next, by Lipschitzness of $\nabla V$,
\begin{align*}
\|\Delta_x\|_{L^p} &\leq L\|y-\widetilde y\|_{L^p} +\left\|\left\|\left(D\mathcal G-\nabla^2V\right)[\widetilde y,x]\right\|_{\op}\|\widetilde y-x\|\right\|_{L^p},
\end{align*}
where
\[
\left(D\mathcal G-\nabla^2V\right)[\widetilde y,x]
:=\int_0^1
\left(D\mathcal G-\nabla^2V\right)(\widetilde y+s(x-\widetilde y))\,ds.
\]
From \eqref{eq:refined_y_coupling} and $Z_d-Z'_d=A_x-D_x$,
\[
\|y-\widetilde y\|_{L^p} \le \eta\mathcal F(h/2) \left(W_p(\mu_x*\g,\g)+h\eta\sigma_p(x)\right).
\]
Also, $\widetilde y-x$ has the same law as $Y_x-x=\int_0^{h/2}V_s\,ds$, where $(V_s)$ is a stationary OU velocity with $V_s\sim\mathcal N(0,I_d)$. Thus, for $p\in\{1,2\}$, we have
\[
\|\widetilde y-x\|_{L^p}\leq\|\widetilde y-x\|_{L^2} \leq\int_0^{h/2}\|V_s\|_{L^2}\,ds =\frac{h}{2}\sqrt d.
\]
Because the construction preserves $\widetilde y\perp\!\!\!\perp\omega$, Assumption~\ref{Assumption:Bounded_Variance} and Jensen's inequality imply
\[
\mathbb{E}\left[
\left\|\left(D\mathcal G-\nabla^2V\right)[\widetilde y,x]\right\|_{\op}^2
\,\middle|\,\widetilde y
\right]
\leq C_G.
\]
Consequently, conditioning on $\widetilde y$ and using Cauchy--Schwarz (directly for $p=2$, and after one further Cauchy--Schwarz inequality for $p=1$),
\[
\left\|\left\|\left(D\mathcal G-\nabla^2V\right)[\widetilde y,x]\right\|_{\op}\|\widetilde y-x\|\right\|_{L^p}\leq\frac{h}{2} C_G^{1/2}\sqrt d.
\]
Therefore
\begin{equation}\label{eq:delta_refined_bound}
\|\Delta_x\|_{L^p}\leq L\eta\mathcal F(h/2)\left(W_p(\mu_x*\g,\g)+h\eta\sigma_p(x)\right)+\frac{h}{2}C_G^{1/2}\sqrt d.
\end{equation}

Substituting \eqref{eq:delta_refined_bound} into the preceding estimate gives three contributions. For the convolution term,
\[
q_0+h q_1L\eta\mathcal F(h/2) \leq \frac{17}{16\sqrt L}+h\sqrt L\,\eta\mathcal F(h/2)\leq\frac{69}{64\sqrt L},
\]
where the last step uses $Lh^2\le1/32$ and $\mathcal F(h/2)\leq h/2$. The Jacobian-noise contribution is bounded by
\[
hq_1\frac{h}{2}C_G^{1/2}\sqrt d \leq \frac{h^2C_G^{1/2}\sqrt d}{2\sqrt L}.
\]
For the two $\sigma_p(x)$ contributions, again let $f=1-\eta$ and $r=L/\gamma^2$. Since $\eta\ge e^{-1/4}$,
\[
|1-\eta-\eta^2|=\eta+\eta^2-1=1-3f+f^2.
\]
Moreover, $f/(\gamma h) \leq 1/2$ and $\sqrt r\le1/(2\sqrt2)$, and hence
\begin{align*}
\sqrt r\,\eta^2f+\frac{f}{\gamma h}(1-3f+f^2) &\leq \frac{f}{2\sqrt2}+\frac12(1-3f+f^2)\leq\frac{1}{2},
\end{align*}
the last inequality following from $0\leq f\leq 1$. Therefore
\begin{equation}\label{eq:local_twisted_sharp}
W_{p,a,b}(\eta_x,\eta_x^{SG})\leq \frac{h^2C_G^{1/2}\sqrt d}{2\sqrt L}+\frac{h^2}{2}\sigma_p(x)+\frac{69}{64}\frac{W_p(\mu_x*\g,\g)}{\sqrt L}.
\end{equation}
Taking the $L^p(\pi)$ norm in $x$ yields
\[
W_{p,a,b}(\overline\pi P,\overline\pi P_{\mathrm{SG}})\leq\frac{h^2C_G^{1/2}\sqrt d}{2\sqrt L}+\frac{h^2}{2}\sigma_p+\frac{69}{64}\frac{\left(\mathbb{E}_{x\sim\pi}W_p^p(\mu_x*\g,\g)\right)^{1/p}}{\sqrt L}.
\]
Substituting this into \eqref{eq:prelocal} gives
\begin{align*}
W_{p,a,b}(\overline\pi,\overline\pi_h)&\leq\frac{\gamma Lh}{mL-\frac{32}{7}hC_G\gamma}\Bigg[\frac{24}{7}\sqrt d(\sqrt L+\gamma)+\frac{4C_G^{1/2}\sqrt d}{\sqrt L}+4\sigma_p\\
&\hspace{8em}+\frac{69}{8}\frac{\left(\mathbb{E}_{x\sim\pi}W_p^p(\mu_x*\g,\g)\right)^{1/p}}{h^2\sqrt L}\Bigg]\\
&\leq\frac{\gamma Lh}{mL-\frac{32}{7}hC_G\gamma}\left[\frac{24}{7}\sqrt d(\sqrt L+\gamma)+\frac{4C_G^{1/2}\sqrt d}{\sqrt L}+4\sigma_p\right.\\
&\hspace{8em}\left.{}+9\frac{\left(\mathbb{E}_{x\sim\pi}W_p^p(\mu_x*\g,\g)\right)^{1/p}}{h^2\sqrt L}\right],
\end{align*}
which is exactly \eqref{eq:sg_ubu_bias}.
\end{proof}

\subsection{Proof of Corollary~\ref{cor:plug_in}}

\pluginBounds*

\begin{proof}[Proof of Corollary~\ref{cor:plug_in}]
For each $x\in\R^d$, recall that
\[
\mu_x=\mathrm{Law}\left(he^{-h\gamma/2}(\nabla V(x)-\mathcal G(x,\omega))\right).
\]
\noindent\textbf{(i)}
Since
\[
\int \|y\|^{2p}\mu_x(dy)
=
h^{2p}e^{-ph\gamma}\E_\omega\|\nabla V(x)-\mathcal G(x,\omega)\|^{2p},
\]
applying a Gaussian convolution inequality pointwise in $x$ gives
\[
W_p(\mu_x*\g,\g)
\le
k_p\, h^2e^{-h\gamma}
\left(\E_\omega\|\nabla V(x)-\mathcal G(x,\omega)\|^{2p}\right)^{1/p},
\]
where for $p=1$ we use Theorem~\ref{thm:conv_W1}, with $k_1=1/\sqrt{2\pi}$, and for $p=2$ we use Theorem~\ref{thm:stochastic-localization-wp}, with $k_2=\sqrt{2-1}/2=1/2$. Taking the $L^p(\pi)$ norm in $x$ yields
\[
\frac{9(\E_{x\sim\pi}W_p^p(\mu_x*\g,\g))^{1/p}}{h^2\sqrt L}\leq \frac{9k_p\,\sigma_{2p}^2}{\sqrt L}=\frac{c_p\,\sigma_{2p}^2}{\sqrt L}.
\]

\noindent\textbf{(ii)} The law $\mu_x$ has Poincar\'e constant $h^2e^{-h\gamma}C_P(x)$ and covariance $h^2e^{-h\gamma}\Cov(\mathcal R(x,\cdot))$, where $\mathcal R(x,\omega):=\nabla V(x)-\mathcal G(x,\omega)$. Theorem~\ref{thm:conv_poincare} gives, for $p\in\{1,2\}$,
\[
W_p(\mu_x*\g,\g)\le \frac12 h^2e^{-h\gamma}\sqrt{C_P(x)\tr(\Cov(\mathcal R(x,\cdot)))},
\]
and taking the $L^p(\pi)$ norm in $x$ (with Cauchy--Schwarz for $p=1$) gives the claim.

\noindent\textbf{(iii)--(iv)} Let $A_x:=he^{-h\gamma/2}\mathcal R(x,\omega)\sim\mu_x$. By Theorem~\ref{thm:conv_refined_12} and $\min(1,\tau)\leq\tau$,
\[
W_p(\mu_x*\g,\g)
\le
2\tau_p(x)+\sqrt{2\log(1+c_4\|\widetilde\Sigma_x\|_F^2)},
\qquad p\in\{1,2\},
\]
with $\tau_p(x)$ and $\widetilde\Sigma_x$ the tail term and truncated covariance of $\mu_x$. Since $A_x$ is centred, $u^\top \widetilde\Sigma_x u=\Var(\langle u,A_x\mathbf 1_{\{\|A_x\|\leq 1\}}\rangle)\le u^\top \Cov(A_x)u$ for all $u$, so $\widetilde\Sigma_x\preceq\Cov(A_x)$ and
\[
\sqrt{2\log(1+c_4\|\widetilde\Sigma_x\|_F^2)}
\le\sqrt{2c_4}\,\|\widetilde\Sigma_x\|_F
\le
\sqrt{2c_4}\,h^2e^{-h\gamma}\sqrt d\,
\lambda_{\max}(\Cov(\mathcal R(x,\cdot))).
\]
Integrating in $x$ for $p=1$, or taking the $L^2(\pi)$ norm for $p=2$, then multiplying by $9/(h^2\sqrt L)$ and using $9\sqrt{2c_4}\leq 23$ proves (iii) and (iv).
\end{proof}

\section{Convergence of the UBU integrator with stochastic gradients}\label{app:ubu_conv}
The full-gradient contraction in Proposition~\ref{prop:Wasserstein} is \cite[Proposition D.6]{ububu}. Here we prove its stochastic-gradient extension.

\sgubuWasserstein*

\begin{proof}[Proof of Proposition~\ref{prop:wass_SGUBU}]
For stochastic gradients, synchronously couple both the Brownian increments and the stochastic-gradient variables. Fix one step and write
\[
\delta x:=\Tilde{x}_{k}-x_{k},\qquad \delta v:=\Tilde{v}_{k}-v_{k},\qquad \delta z:=(\delta x,\delta v).
\]
Let $\eta=e^{-\gamma h/2}$ and $\alpha=(1-\eta)/\gamma$, then conditional on the current states and on the common Gaussian randomness in the first $\mathcal U$ half-step, the difference between the two positions at which the gradient is evaluated is
\[
\delta y:=\Tilde{y}-y=\delta x+\alpha\delta v.
\]
Define the centred stochastic-gradient error
\[
\varepsilon_{\omega} := \left(\mathcal G(\Tilde{y},\omega)-\mathcal G(y,\omega)\right) - \left(\nabla V(\Tilde{y})-\nabla V(y)\right).
\]
Pointwise unbiasedness of the stochastic gradient gives $\mathbb{E}_{\omega}[\varepsilon_{\omega}]=0$. Moreover, we have
\[
\varepsilon_{\omega} = \int_{0}^{1}\left(D_x\mathcal G(y+t\delta y,\omega)-\nabla^{2}V(y+t\delta y)\right)\delta y\,dt,
\]
and by Jensen's inequality, Fubini's theorem and Assumption~\ref{Assumption:Bounded_Variance},
\begin{align*}
\E_{\omega}\|\varepsilon_{\omega}\|^{2} &\leq\int_{0}^{1}\E_{\omega}\left\|\left(D_x\mathcal G(y+t\delta y,\omega)-\nabla^{2}V(y+t\delta y)\right)\delta y \right\|^{2}dt\\
&\leq C_G\|\delta y\|^{2}.
\end{align*}

Now let $\delta z_{+}^{\mathrm{FG}}$ denote the output difference for the corresponding full-gradient UBU step, with the same input states and common Gaussian variables, and let $\delta z_{+}^{\mathrm{SG}}$ denote the stochastic-gradient output difference. The gradient error first produces a velocity error $-h\varepsilon_{\omega}$ and is then propagated through the second $\mathcal U$ half-step to give
\[
\delta z_{+}^{\mathrm{SG}} = \delta z_{+}^{\mathrm{FG}}-hS\varepsilon_{\omega}, \qquad S:= \begin{pmatrix}
\alpha I_d\\
\eta I_d
\end{pmatrix}.
\]
Introduce the matrix associated with the twisted norm,
\[
\mathcal M:=
\begin{pmatrix}
I_d & bI_d\\
bI_d & aI_d
\end{pmatrix}, \qquad \|z\|_{a,b}^{2}=z^{T}\mathcal Mz.
\]
Since $\delta z_{+}^{\mathrm{FG}}$ is conditionally independent of $\omega$ and $\E_{\omega}\varepsilon_{\omega}=0$, the cross term vanishes and
\[
\mathbb{E}_{\omega}\|\delta z_{+}^{\mathrm{SG}}\|_{a,b}^{2}
=
\|\delta z_{+}^{\mathrm{FG}}\|_{a,b}^{2}
+h^{2}\E_{\omega}\left[\varepsilon_{\omega}^{T}S^{T}\mathcal MS\varepsilon_{\omega}\right].
\]
A direct calculation, using $\alpha=b(1-\eta)$, gives
\[
S^{T}\mathcal MS = \left(\alpha^{2}+2b\alpha\eta+a\eta^{2}\right)I_d = \left(b^{2}(1-\eta^{2})+a\eta^{2}\right)I_d \preceq aI_d,
\]
where the final inequality follows from $b^{2}\leq a$, which is implied by $\gamma^{2}\geq 8L$. Consequently,
\[
\E_{\omega}\|\delta z_{+}^{\mathrm{SG}}\|_{a,b}^{2} \leq \|\delta z_{+}^{\mathrm{FG}}\|_{a,b}^{2}+ah^{2}C_G\|\delta x+\alpha\delta v\|^{2}.
\]

It remains to control the intermediate-position difference, $\delta y =\delta x+\alpha\delta v$, by the twisted norm. Let $T=(I_d,\alpha I_d)$, then since
\[
\mathcal M^{-1} = \frac{1}{a-b^{2}} \begin{pmatrix}
aI_d&-bI_d\\
-bI_d&I_d
\end{pmatrix},
\]
Cauchy--Schwarz in the $\mathcal M$-inner product gives
\begin{align*}
\|\delta x+\alpha\delta v\|^{2} &\leq \|T\mathcal M^{-1/2}\|_{\op}^{2}\|\delta z\|_{a,b}^{2}= \frac{a-2b\alpha+\alpha^{2}}{a-b^{2}}\|\delta z\|_{a,b}^{2}\\
&\leq \frac{a}{a-b^{2}}\|\delta z\|_{a,b}^{2} = \frac{1}{1-L/\gamma^{2}}\|\delta z\|_{a,b}^{2} \leq \frac{8}{7}\|\delta z\|_{a,b}^{2}.
\end{align*}
Here we used $0\leq\alpha\leq b$, so that $a-2b\alpha+\alpha^{2}\leq a$, and $L/\gamma^{2}\leq 1/8$. Combining the preceding estimates with the full-gradient contraction from Proposition~\ref{prop:Wasserstein} gives
\begin{align*}
\E_{\omega}\|\delta z_{+}^{\mathrm{SG}}\|_{a,b}^{2}
&\leq
\left(1-\frac{mh}{4\gamma}\right)\|\delta z\|_{a,b}^{2}
+\frac{8}{7}ah^{2}C_G\|\delta z\|_{a,b}^{2}\\
&=
\left(1-\frac{mh}{4\gamma}+\frac{8h^{2}C_G}{7L}\right)
\|\delta z\|_{a,b}^{2},
\end{align*}
and we have the first claim and Wasserstein contraction. 

Existence of a unique invariant measure follows from the additional moment condition in Assumption~\ref{Assumption:Bounded_Variance} and the Jacobian bound that provide the $L^{2}$ linear-growth estimate
\[
\left(\E_{\omega}\|\mathcal G(x,\omega)\|^{2}\right)^{1/2} \leq \left(\E_{\omega}\|\mathcal G(x_{\star},\omega)\|^{2}\right)^{1/2}+\left(L+\sqrt{C_G}\right)\|x-x_{\star}\|.
\]
Hence the SG-UBU transition kernel maps $\mathcal P_2(\mathbb{R}^{2d})$ into itself and further if
\[
h<\min\left\{\frac{1}{2\gamma},\frac{7mL}{32C_G\gamma}\right\},
\]
then $1-mh/(4\gamma)+8h^{2}C_G/(7L)<1$. Thus $\mu\mapsto\mu P_h$ is a strict contraction in $(\mathcal P_2(\mathbb{R}^{2d}),W_{2,a,b})$, and the Banach fixed-point theorem yields a unique invariant measure $\overline\pi_h\in\mathcal P_2(\mathbb{R}^{2d})$.
\end{proof}

\section{A Gaussian lower bound example}
\label{sec:gaussian-sharpness}

Consider the standard Gaussian target with additive Gaussian stochastic-gradient noise,
\[
V(x)=\frac12\|x\|^2,
\qquad
\mathcal G(x,\xi)=x+\tau\xi,
\qquad
\xi\sim\mathcal N(0,I_d),
\]
where an independent copy of $\xi$ is used at every step.  The parameter $\tau^2$ is the noise variance in each coordinate.  Since the SG-UBU update acts coordinatewise, the invariant covariance is determined by a two-dimensional calculation.

\begin{lemma}\label{lem:gaussian-invariant-law}
Suppose that
\[
h\tanh(\gamma h/2)<2\gamma.
\]
Then SG-UBU has a unique invariant distribution of the form
\[
\overline\pi_{h,\tau}
=\mathcal N\left(0,\mathsf\Sigma_{h,\tau}\otimes I_d\right).
\]
For every fixed $\gamma>0$ and $\tau\geq0$, as $h\downarrow0$,
\begin{equation}\label{eq:gaussian-covariance-expansion}
\mathsf\Sigma_{h,\tau}
=I_2+\frac{\tau^2}{2\gamma}hI_2
+h^2
\begin{pmatrix}
-1/6&\gamma/24\\[1mm]
\gamma/24&1/12
\end{pmatrix}
+\mathcal{O}(h^3).
\end{equation}
\end{lemma}

\begin{proof}
It is enough to work in one coordinate.  Set
\[
\eta=e^{-\gamma h/2},
\qquad
\alpha=\frac{1-\eta}{\gamma},
\qquad
\mathsf U_h=
\begin{pmatrix}1&\alpha\\0&\eta\end{pmatrix},
\qquad
\mathsf B_h=
\begin{pmatrix}1&0\\-h&1\end{pmatrix},
\qquad
\mathsf c_h=
\begin{pmatrix}\alpha\\\eta\end{pmatrix}.
\]
The noise in one exact $\mathcal U$ half-step has covariance
\[
\mathsf R_h=
\begin{pmatrix}
\dfrac{h}{\gamma}-\dfrac{3-4\eta+\eta^2}{\gamma^2}
&\dfrac{(1-\eta)^2}{\gamma}\\[3mm]
\dfrac{(1-\eta)^2}{\gamma}
&1-\eta^2
\end{pmatrix}.
\]
Thus, with
\[
\mathsf A_h=\mathsf U_h\mathsf B_h\mathsf U_h,
\qquad
\mathsf Q_h
=\mathsf U_h\mathsf B_h\mathsf R_h
 \mathsf B_h^{\top}\mathsf U_h^{\top}+\mathsf R_h,
\]
one step has the representation
\begin{equation}\label{eq:gaussian-linear-chain}
Z_{k+1}
=\mathsf A_hZ_k+\zeta_{k+1}-h\tau\mathsf c_h\xi_{k+1},
\qquad
Z_k=(X_k,V_k)^\top,
\end{equation}
where $\zeta_{k+1}\sim\mathcal N(0,\mathsf Q_h)$ and $\xi_{k+1}\sim\mathcal N(0,1)$ are independent of each other and of the past.

Now
\[
\det(\mathsf A_h)=\eta^2,
\qquad
\operatorname{tr}(\mathsf A_h)
=1+\eta^2-\frac{h}{\gamma}(1-\eta^2).
\]
Since $0<\eta<1$, the stability criterion gives
\[
\rho(\mathsf A_h)<1
\quad\Longleftrightarrow\quad
h\tanh(\gamma h/2)<2\gamma.
\]
In this regime, the covariance of the unique invariant law is given exactly by the convergent series
\begin{equation}\label{eq:gaussian-stationary-cov}
\mathsf\Sigma_{h,\tau}
=\sum_{j=0}^{\infty}
\mathsf A_h^j
\left(\mathsf Q_h+h^2\tau^2\mathsf c_h\mathsf c_h^\top\right)
(\mathsf A_h^\top)^j,
\end{equation}
or equivalently,
\[
\mathsf\Sigma_{h,\tau}
=\mathsf A_h\mathsf\Sigma_{h,\tau}\mathsf A_h^\top
 +\mathsf Q_h+h^2\tau^2\mathsf c_h\mathsf c_h^\top.
\]
Solving the three scalar equations in this identity and expanding at $h=0$ gives \eqref{eq:gaussian-covariance-expansion}.
\end{proof}

The expansion illustrates the source of the first-order error: full-gradient UBU perturbs the covariance only at order $h^2$, whereas stochastic-gradient noise produces an order $h$ covariance inflation.

\begin{proposition}
\label{prop:gaussian-exact-stationary}
Fix $\gamma>0$ and $\tau>0$, and let $\pi^x_{h,\tau}$ denote the position marginal of $\overline\pi_{h,\tau}$.  Then, as $h\downarrow0$,
\begin{align}
W_2\left(\mathcal N(0,I_d),\pi^x_{h,\tau}\right)
&=\frac{\tau^2}{4\gamma}h\sqrt d
+\mathcal{O}(h^2\sqrt d),
\label{eq:gaussian-position-sharp}\\
W_2\left(\mathcal N(0,I_{2d}),\overline\pi_{h,\tau}\right)
&=\frac{\tau^2}{2\sqrt2\,\gamma}h\sqrt d
+\mathcal{O}(h^2\sqrt d).
\label{eq:gaussian-extended-sharp}
\end{align}
Consequently, for fixed $\gamma\geq 2\sqrt2$ and $\tau^2>0$, the $h\sqrt d$ dependence in Theorem~\ref{theorem:bias_sgUBU} is sharp up to constant factors.
\end{proposition}

\begin{proof}
Let $s_{h,\tau}=(\mathsf\Sigma_{h,\tau})_{11}$, then by \eqref{eq:gaussian-covariance-expansion},
\[
s_{h,\tau} =1+\frac{\tau^2}{2\gamma}h-\frac{1}{6}h^2 +\mathcal{O}(h^3).
\]
Since $\pi^x_{h,\tau}=\mathcal N(0,s_{h,\tau}I_d)$,
\[
W_2\left(\mathcal N(0,I_d),\pi^x_{h,\tau}\right)
=\sqrt d\,\bigl|\sqrt{s_{h,\tau}}-1\bigr|,
\]
and expanding the square root gives \eqref{eq:gaussian-position-sharp}.

For the extended law,
\[
W_2^2\left(\mathcal N(0,I_{2d}),\overline\pi_{h,\tau}\right)
=d\,\operatorname{tr}
\left(I_2+\mathsf\Sigma_{h,\tau}
      -2\mathsf\Sigma_{h,\tau}^{1/2}\right).
\]
The matrix square-root expansion based on
\eqref{eq:gaussian-covariance-expansion} gives
\eqref{eq:gaussian-extended-sharp}.
\end{proof}

Here $m=L=1$, $C_G=0$ and $\sigma_1\leq\sigma_2=\tau\sqrt d$, and the gradient noise $\mathcal R(x,\omega)=-\tau\xi$ is Gaussian, so it satisfies a Poincar\'e inequality with $C_P(x)=\tau^2$ and $\mathrm{tr}(\mathrm{Cov}(\mathcal R(x,\cdot)))=\tau^2 d$. Hence, for $\gamma\geq 2\sqrt 2$ (so that $\gamma\geq\sqrt{8L}$), Theorem~\ref{theorem:bias_sgUBU} with Corollary~\ref{cor:plug_in}(ii), together with the norm equivalence in Remark~\ref{rem:norm_equiv}, gives an $\mathcal{O}(h\sqrt d)$ upper bound, matching \eqref{eq:gaussian-position-sharp}-\eqref{eq:gaussian-extended-sharp} in its dependence on $h$ and $d$. The same holds in $W_1$: since $W_1(\mathcal N(0,I_d),\mathcal N(0,sI_d))=|\sqrt s-1|\,\E\|Z_d\|$ for $Z_d\sim\mathcal N(0,I_d)$ (the coupling $Z_d\mapsto\sqrt s Z_d$ gives the upper bound and the $1$-Lipschitz function $x\mapsto\|x\|$ the lower bound), and $\E\|Z_d\|/\sqrt d\to1$, the position marginal satisfies $W_1(\mathcal N(0,I_d),\pi^x_{h,\tau})=\frac{\tau^2}{4\gamma}h\,\E\|Z_d\|+\mathcal{O}(h^2\sqrt d)$.

\section{A functional-inequality bound for \texorpdfstring{$W_1(\mu*\g,\g)$ and $W_2(\mu*\g,\g)$}{W1(mu*g,g) and W2(mu*g,g)}}\label{sec:functional}
Let $\g:=\mathcal N(0, I_d)$, let $X\sim\mu$ be centred with $\mathbb{E}\|X\|^2<\infty$, let $Z\sim\g$ be independent of $X$, and set $\nu := \mu*\g=\mathrm{Law}(X+Z)$. With $\tau_1,\tau_2$ as in Theorem~\ref{thm:conv_refined_12}, write $X_{\leq 1}:=X\mathbf 1_{\{\|X\|\leq 1\}}$, $X_{>1}:=X-X_{\leq 1}$ and $m := \mathbb{E}[X_{\leq 1}]=-\mathbb{E}[X_{>1}]$. Then by Jensen's inequality
\begin{equation}\label{eq:m-le-1}
\|m\|\leq \min(1,\tau_1),\qquad \|m\|\leq \min(1,\tau_2).
\end{equation}
Let $\widetilde X := X_{\leq 1}-m$, $\widetilde\nu := \mathrm{Law}(Z+\widetilde X)$ and $\widetilde\Sigma := \mathrm{Cov}(\widetilde X)$.

\begin{lemma}\label{lem:trunc-error}
For $p\in\{1,2\}$, $W_p(\nu,\widetilde\nu)\leq \tau_p+\min(1,\tau_p)$.
\end{lemma}

\begin{proof}
Using the same $(Z,X)$ for both laws, $(Z+X)-(Z+\widetilde X)=X_{>1}+m$, so $W_p(\nu,\widetilde\nu)\leq\|X_{>1}\|_{L^p}+\|m\|=\tau_p+\|m\|$. Then using \eqref{eq:m-le-1} we have the required result.
\end{proof}

\begin{lemma}\label{lem:KL}
Let $c_4:=(e^4-5)/16$, then
\[
\mathrm{KL}(\widetilde\nu\|\g)
\leq \log\left(1+\chi^2(\widetilde\nu\|\g)\right)
\leq \log\left(1+c_4\|\widetilde\Sigma\|_F^2\right),
\]
where $\mathrm{KL}(\rho\|\g)=\int\log\frac{d\rho}{d\g}d\rho$ and $\chi^2(\rho\|\g)=\int \big(\frac{d\rho}{d\g}-1\big)^2d\g$.
\end{lemma}

\begin{proof}
The first inequality is the standard bound $\mathrm{KL}\leq\log(1+\chi^2)$. Let $\widetilde X',\widetilde X''$ be i.i.d.\ copies of $\widetilde X$. Since $\frac{d\widetilde\nu}{d\g}(y)=\mathbb{E}\exp\big(\langle y,\widetilde X'\rangle-\tfrac12\|\widetilde X'\|^2\big)$, expanding the square, applying Fubini and using $\int e^{\langle y,u\rangle}d\g(y)=e^{\|u\|^2/2}$ gives
\[
\chi^2(\widetilde\nu\|\g)=\mathbb{E}\big[e^{T}\big]-1,
\qquad T:=\langle \widetilde X',\widetilde X''\rangle.
\]
Since $\|X_{\leq 1}\|\leq 1$ and $\|m\|\leq 1$ by \eqref{eq:m-le-1}, $\|\widetilde X\|\leq 2$ and hence $|T|\leq 4$ a.s. The function $t\mapsto(e^t-1-t)/t^2$ is increasing, so $e^t\leq 1+t+c_4t^2$ on $[-4,4]$. As $\mathbb{E}T=0$ and $\mathbb{E}T^2=\sum_{i,j}(\mathbb{E}[\widetilde X_i\widetilde X_j])^2=\|\widetilde\Sigma\|_F^2$, we obtain $\chi^2(\widetilde\nu\|\g)\leq c_4\|\widetilde\Sigma\|_F^2$.
\end{proof}

Combining Lemma~\ref{lem:KL} with Talagrand's inequality $W_2(\g,\rho)\leq\sqrt{2\mathrm{KL}(\rho\|\g)}$ and $W_1\leq W_2$ yields
\begin{equation}\label{eq:W12-gamma-tilde}
W_1(\g,\widetilde\nu)\leq W_2(\g,\widetilde\nu)\leq \sqrt{2\log\left(1+c_4\|\widetilde\Sigma\|_F^2\right)}.
\end{equation}

\convRefined*

\begin{proof}
For $p\in\{1,2\}$, the triangle inequality gives
$W_p(\g,\nu)\leq W_p(\g,\widetilde\nu)+W_p(\widetilde\nu,\nu)$. We then bound the first term by \eqref{eq:W12-gamma-tilde} and the second by Lemma~\ref{lem:trunc-error} to establish the required bounds.
\end{proof}

\subsection{A spike example showing that a covariance-Frobenius bound cannot hold in general}\label{sec:spike}

The following example shows that one cannot, in general control $W_1(\mu * \g, \g)$ uniformly by an absolute constant
times $\|\Sigma\|_F$ without an additional tail term.
\begin{proposition}\label{prop:spike-lower-bound}
Fix $d\geq 2$, $s>0$, and define
\[
\mu_{s,d}=\frac1{2d}\sum_{i=1}^d\left(\delta_{s e_i}+\delta_{-s e_i}\right).
\]
Let $Z\sim\g=\mathcal N(0,I_d)$ and $X\sim\mu_{s,d}$ be independent, and set
$
\nu:=\mu_{s,d} * \g=\mathrm{Law}(Z+X).
$
Then the covariance of $\mu_{s,d}$ is
\[
\Sigma=\mathrm{Cov}(\mu_{s,d})=\frac{s^2}{d}I_d,
\qquad
\|\Sigma\|_F=\frac{s^2}{\sqrt d},
\]
and
\begin{equation}\label{eq:spike-lower}
W_1(\g,\nu)\ \ge\ \frac{s}{2}-\sqrt{2\log d}.
\end{equation}
In particular, if $s\geq 4\sqrt{2\log d}$, then
\begin{equation}\label{eq:spike-lower-clean}
W_1(\mu_{s,d} * \g, \g)\ \ge\ \frac{s}{4}.
\end{equation}
\end{proposition}

\begin{proof}
The covariance follows from symmetry: $\mathbb{E}[X]=0$ and $\mathbb{E}[XX^\top]=\frac{1}{2d}\sum_{i}2s^2e_ie_i^\top=\frac{s^2}{d}I_d$. For the lower bound, $f(x):=\max_{i}x_i$ is $1$-Lipschitz, so by Kantorovich--Rubinstein duality
\begin{equation}\label{eq:KR-spike}
W_1(\g,\nu)\geq \mathbb{E}[f(Z+X)]-\mathbb{E}[f(Z)].
\end{equation}
Write $X=s\varepsilon e_I$ with $I$ uniform on $\{1,\dots,d\}$ and $\varepsilon$ a Rademacher sign, independent of each other and of $Z$. On $\{\varepsilon=+1\}$, $f(Z+X)\geq Z_I+s$, and on $\{\varepsilon=-1\}$, $\mathbb{E}[f(Z+X)\mid\varepsilon=-1]\geq\mathbb{E}[\max_{j\neq I}Z_j]\geq 0$ since $d\geq 2$; hence $\mathbb{E}[f(Z+X)]\geq s/2$. On the other hand, for $\lambda>0$, Jensen's inequality gives $\mathbb{E}[f(Z)]\leq\lambda^{-1}\log\sum_i\mathbb{E}e^{\lambda Z_i}=\lambda^{-1}\log d+\lambda/2$, and $\lambda=\sqrt{2\log d}$ yields $\mathbb{E}[f(Z)]\leq\sqrt{2\log d}$. This proves \eqref{eq:spike-lower}, and \eqref{eq:spike-lower-clean} follows when $\sqrt{2\log d}\leq s/4$.
\end{proof}

\begin{corollary}\label{cor:spike-no-frob-only}
There is no absolute constant $C<\infty$ such that
\[
W_1(\mu * \g, \g)\leq C\|\Sigma\|_F
\]
for every centred probability measure $\mu$ with finite second moment.
\end{corollary}

\begin{proof}
Take $s=d^{1/4}$ in Proposition~\ref{prop:spike-lower-bound}. Then $\|\Sigma\|_F=1$, while for $d$ large enough $s\geq 4\sqrt{2\log d}$ and \eqref{eq:spike-lower-clean} gives $W_1(\mu_{s,d} * \g, \g)\geq d^{1/4}/4\to\infty$.
\end{proof}

\begin{remark}
This example shows that in general the covariance/Frobenius term by itself cannot control $W_1(\mu_{s,d} * \g, \g)$. In particular, the tail contribution in Theorem~\ref{thm:conv_refined_12} is necessary.
\end{remark}

\subsection{Heuristic and numerical evidence for SG-UBU with spike stochastic-gradient noise}\label{sec:spike:SGUBU}
We give a heuristic and numerical evidence that, under only finite second moments, the stationary bias of SG-UBU can scale substantially worse in $d$ than the $\mathcal O(h\sqrt d)$ rate available under Poincar\'e or bounded-noise assumptions; Appendix~\ref{sec:spike} only concerns the one-step convolution inequality. Consider the noise distribution $\mu_{s,d,p}=(1-p)\delta_{0} + p\mu_{s,d}$ with $0\leq p\leq 1$.

Let $f_{k}(x) = \max_{|S|=k} \|x_S\|_2$, where $x_S \in \mathbb{R}^d$ denotes the vector obtained from $x$ by zeroing out all coordinates outside $S$. Since $f_{k}(x)$ is a pointwise maximum of the coordinate-projection norms $x \mapsto \|x_S\|_2$, each of which is $1$-Lipschitz with respect to the Euclidean norm, $f_k$ is itself $1$-Lipschitz.

Suppose that the true target is a $d$-dimensional standard Gaussian. For the following heuristic, when using SG-UBU with stepsize $h$, we consider the following choice of parameters: $p=h \sqrt{d}/\log(d)$, $s=(8/h)\sqrt{\log(d)}$. Consider the case when $h=d^{-1/2}$. SG-UBU on a standard Gaussian target with i.i.d. noise distribution $\mu_{s,d,p}$ would essentially consist of shocks that are added to the velocity distribution with probability $p$ each step, and the size of these would be $h\cdot s=8\sqrt{\log(d)}$, requiring $\mathcal{O}(1/h)$ steps to approximately recover from them. During a unit diffusion time period, the expected number of shocks is $p/h=\sqrt{d}/\log(d)$. Hence if we consider the top $k=\lceil \sqrt{d}/\log(d)\rceil$ largest magnitude elements of the position, the $L_2$ norm of this subvector is expected to be of size $\mathcal{O}\left(\sqrt{\log(d)} \left(\sqrt{d}/\log(d)\right)^{1/2}\right)=\mathcal{O}\left(d^{1/4}\right)$. For the original standard Gaussian target, one can show with a simple argument that the expected value of $f_k$ with $k=\lceil \sqrt{d}/\log(d)\rceil$ is $\mathcal{O}(d^{1/4})$. Hence, the difference between these is expected to be also of $\mathcal{O}(d^{1/4})$.

\paragraph{Numerical estimation via synchronously coupled chains.}
We estimate the stationary bias of $f_{k}$ by the long-run average of $f_{k}(x^{\mathrm{SG}})-f_{k}(x^{\mathrm{UBU}})$ along an SG-UBU chain and a full-gradient UBU chain driven by the same Gaussian increments $\xi^{(1)},\xi^{(2)}$ in every $\mathcal{U}$ half-step and started from the same draw $x_{0},v_{0}\sim \mathcal{N}(0,I_{d})$. For a Gaussian target the shared increments cancel in the difference, which greatly reduces the variance. The estimator targets $\E_{\overline{\pi}_{h}^{\mathrm{SG}}}[f_{k}]-\E_{\overline{\pi}_{h}^{\mathrm{UBU}}}[f_{k}]$, which differs from the bias relative to $\pi$ by the $\mathcal{O}(h^{2})$ bias of UBU; we verified numerically that this difference is of order $10^{-3}$ here. We use $\gamma=2$ and $128$ independent replicas (\texttt{spike\_bias.ipynb} in \texttt{\url{https://github.com/PAWhalley/SG-UBU}}). For comparison, a third chain with Gaussian gradient noise $\mathcal{N}(0,(ps^{2}/d) I_{d})$, matching the covariance of $\mu_{s,d,p}$, is coupled to the same increments.

Table~\ref{table:spike:res} shows the resulting bias estimates for $h=d^{-1/2}$. They are consistent with the heuristic $\mathcal O(d^{1/4})$ growth: the fitted log-log slope of the spike-noise column is $0.251$, in close agreement with the predicted exponent $1/4$, whereas the Gaussian-noise bias remains bounded in $d$.

\begin{table}[ht]
\centering
\setlength{\tabcolsep}{5pt}
\begin{tabular}{rcc}
\toprule
Dimension $d$ & Bias of $f_k$ for spike gradient noise & Bias of $f_k$ for Gaussian gradient noise\\
\midrule
64     & 5.010  & 2.512 \\
256    & 6.995  & 2.002 \\
1{,}024  & 10.031 & 1.558 \\
4{,}096  & 14.082 & 1.122 \\
16{,}384 & 20.145 & 0.817 \\
\bottomrule
\end{tabular}
\caption{Bias of 1-Lipschitz test function $f_{k}$ for $k=\lceil \sqrt{d}/\log(d)\rceil$ for SG-UBU (spike/Gaussian stochastic gradient noise distribution, friction parameter $\gamma=2$). Step size $h=d^{-1/2}$, and we keep $h \sqrt{d}\lambda_{\max}(\mathrm{Cov}(\mathcal{R}(x,\omega)))$ constant for all dimensions (i.e. noise covariance is scaled proportionally to $1/(h \sqrt{d})$). Biases are estimated with the coupled-chain procedure described above; standard errors across replicas are at most $0.02$ for the spike column and $0.01$ for the Gaussian column. For Gaussian noise, Theorem~\ref{theorem:bias_sgUBU} and Corollary~\ref{cor:plug_in} imply a uniform bound in $d$.}
\label{table:spike:res}
\end{table}

In summary, the spike experiments support the qualitative message that tail behaviour can affect stationary bias, but all stationary lower-bound claims in this subsection remain heuristic. Establishing a rigorous invariant-measure lower bound for this example is left open.

\section{Bounding \texorpdfstring{$W_2(\mu * \g,\g)$ and $\mathrm{KL}(\mu*\g\|\g)$}{W2(mu*g,g) and KL(mu*g||g)} when \texorpdfstring{$\mu$}{mu} satisfies a Poincar{\'e} inequality}\label{sec:poincare-appendix}

As in Appendix~\ref{sec:functional}, let $X\sim \mu$ be centred with finite second moment and covariance $\Sigma$, and let $Z\sim\g$ be independent of $X$. For $r\in[0,1]$ set
\[
Y_r:=Z+rX,\qquad \rho_r:=\mathrm{Law}(Y_r),
\]
so that $\rho_0=\g$ and $\rho_1=\mu*\g$. For $\rho\ll\g$ we write $\mathrm{KL}(\rho\|\g):=\int\log\frac{d\rho}{d\g}d\rho$. If $W\sim\rho$ is centred, a Stein kernel for $\rho$ is a measurable matrix-valued map $\tau_\rho:\mathbb{R}^d\to\mathbb{R}^{d\times d}$ such that
\[
\E\langle W,\varphi(W)\rangle
=
\E\mathrm{tr}\big(\tau_\rho(W)\nabla\varphi(W)\big)
\qquad\text{for all }\varphi\in C_c^\infty(\mathbb{R}^d;\mathbb{R}^d),
\]
with the same matrix-gradient convention used throughout.

\begin{fact}[Existence of a Stein kernel under a Poincar\'e inequality]\label{fact_1}
By \cite[Theorem 2.4]{courtade2019stein}, if $\mu$ is centred, absolutely
continuous with respect to Lebesgue measure, has finite second moment, and
satisfies \eqref{eq:poincare}, then $\mu$ admits a Stein kernel
$\tau_\mu$ with
\begin{equation}\label{eq:CFP-bound-direct}
\E\|\tau_\mu(X)\|_F^2
\le
C_P\int \|x\|^{2}\mu(dx)
= C_P\mathrm{tr}(\Sigma).
\end{equation}
\end{fact}

\convPoincare*

\convPoincareKL*

\noindent Both theorems follow from Theorem~\ref{thm:poincare-KL} below. The key step is a representation of the posterior mean $\E[X\mid Y_r]$ through the Stein kernel.

\begin{lemma}\label{lem:fisher-conv}
Let $\tau_\mu$ be a Stein kernel for $\mu$ with $B:=\E\|\tau_\mu(X)\|_F^2<\infty$. Then for every $r\in(0,1]$,
\begin{equation}\label{eq:cond-mean-identity}
\E[X\mid Y_r]=r\,\E\big[\tau_\mu(X)Z\mid Y_r\big],
\qquad\text{and hence}\qquad
\E\big\|\E[X\mid Y_r]\big\|^{2}\le r^2B.
\end{equation}
\end{lemma}

\begin{proof}
Write $\phi(z):=(2\pi)^{-d/2}e^{-\|z\|^{2}/2}$, so that $Y_r$ has the strictly positive density $p_r(y)=\E\phi(y-rX)$ and the conditional law of $X$ given $Y_r=y$ is proportional to $\phi(y-rx)\mu(dx)$. Fix $y\in\R^{d}$ and $i\in\{1,\dots,d\}$, and let $\chi_R\in C_c^\infty(\R^{d})$ equal $1$ on $\{\|x\|\leq R\}$, vanish outside $\{\|x\|\leq 2R\}$ and satisfy $\|\nabla\chi_R\|_\infty\leq 2/R$. Since $\nabla_x\phi(y-rx)=r(y-rx)\phi(y-rx)$, the Stein identity of $\tau_\mu$ applied to $x\mapsto\chi_R(x)\phi(y-rx)e_{i}$ reads
\begin{align*}
\E\big[X_i\chi_R(X)\phi(y-rX)\big]
&=r\,\E\big[\chi_R(X)\big(\tau_\mu(X)(y-rX)\big)_i\phi(y-rX)\big]\\
&\quad+\E\big[\big(\tau_\mu(X)\nabla\chi_R(X)\big)_i\phi(y-rX)\big].
\end{align*}
As $R\to\infty$, the first two expectations converge to their versions without $\chi_R$ by dominated convergence, with dominating functions $(2\pi)^{-d/2}|X_i|$ and $c_d\|\tau_\mu(X)\|_F$, where $c_d:=\sup_{u}\|u\|\phi(u)<\infty$; both are integrable since $\E\|X\|<\infty$ and $\E\|\tau_\mu(X)\|_F\leq B^{1/2}<\infty$. The last expectation is at most $(2/R)(2\pi)^{-d/2}\E\|\tau_\mu(X)\|_{F}\to0$. Hence $\E[X\phi(y-rX)]=r\,\E[\tau_\mu(X)(y-rX)\phi(y-rX)]$, and dividing by $p_r(y)$ gives the identity in \eqref{eq:cond-mean-identity}, since $Z=Y_r-rX$. By conditional Jensen, independence of $X$ and $Z$, and $\E[ZZ^{\top}]=I_{d}$,
\[
\E\big\|\E[\tau_\mu(X)Z\mid Y_r]\big\|^{2}
\le\E\big\|\tau_\mu(X)Z\big\|^{2}
=\E\big[\mathrm{tr}\big(\tau_\mu(X)^{\top}\tau_\mu(X)\big)\big]
=B.\qedhere
\]
\end{proof}

\begin{theorem}\label{thm:poincare-KL}
Let $\mu$ be centred with finite second moment, and let $\tau_\mu$ be any Stein kernel for $\mu$ with $B:=\E\|\tau_\mu(X)\|_{F}^{2}<\infty$. Then
\[
W_2(\mu*\g,\g)\ \le\ \frac{\sqrt B}{2}
\qquad\text{and}\qquad
\mathrm{KL}(\mu*\g\|\g)\ \le\ \frac{B}{4}.
\]
In particular, under the assumptions of Theorem~\ref{thm:conv_poincare}, choosing the Stein kernel of Fact~\ref{fact_1} yields $W_2(\mu*\g,\g)\le\frac12\sqrt{C_P\mathrm{tr}(\Sigma)}$ and $\mathrm{KL}(\mu*\g\|\g)\le\frac14C_P\mathrm{tr}(\Sigma)$, which proves Theorems~\ref{thm:conv_poincare} and~\ref{thm:conv_poincare_KL}.
\end{theorem}

\begin{proof}
For the $W_2$ bound, let $v_r(y):=\E[X\mid Y_r=y]$. For $\varphi\in C_c^\infty(\R^d)$, differentiation under the expectation (justified since $\nabla\varphi$ is bounded and $\E\|X\|<\infty$) gives
\[
\frac{d}{dr}\int\varphi\,d\rho_r=\E\langle\nabla\varphi(Y_r),X\rangle=\int\langle\nabla\varphi,v_r\rangle\,d\rho_r,
\]
so $(\rho_r,v_r)$ solves the continuity equation, and $r\mapsto\rho_r$ is narrowly continuous. By Lemma~\ref{lem:fisher-conv}, $\|v_r\|_{L^2(\rho_r)}\leq r\sqrt B$, and Lemma~\ref{lem:ce-length} with $p=2$ gives $W_2(\mu*\g,\g)\leq\int_0^1 r\sqrt B\,dr=\sqrt B/2$.

For the KL bound, since $\E\|Y_r\|^2=d+r^2\E\|X\|^2$ and the differential entropy of $Y_r$ equals $I(X;Y_r)$ plus that of $Z$, where $I$ denotes mutual information, we have
\[
\mathrm{KL}(\rho_r\|\g)=\frac{r^2}{2}\E\|X\|^2-I(X;Y_r).
\]
By the I-MMSE relation of \citet{guo2005mutual}, applied with signal-to-noise ratio $r^2$, $\frac{d}{dr}I(X;Y_r)=r\,\E\|X-\E[X\mid Y_r]\|^2=r\big(\E\|X\|^2-\E\|\E[X\mid Y_r]\|^2\big)$. Hence, by Lemma~\ref{lem:fisher-conv},
\[
\frac{d}{dr}\mathrm{KL}(\rho_r\|\g)=r\,\E\big\|\E[X\mid Y_r]\big\|^2\leq r^3B,
\]
and integrating from $\mathrm{KL}(\rho_0\|\g)=0$ gives $\mathrm{KL}(\mu*\g\|\g)\leq B/4$. The final statements follow from \eqref{eq:CFP-bound-direct}.
\end{proof}

\begin{remark}\label{rem:KL-sharp}
Both constants are attained in the small-noise limit. For Gaussian noise $\mu=\mathcal N(0,\Sigma)$ one may take the constant Stein kernel $\tau_\mu\equiv\Sigma$, so that $B=\|\Sigma\|_F^2$, while, with $\lambda_{1},\dots,\lambda_{d}$ the eigenvalues of $\Sigma$,
\[
\mathrm{KL}\big(\mathcal N(0,I_{d}+\Sigma)\big\|\mathcal N(0,I_{d})\big)
=\frac12\sum_{i=1}^{d}\big(\lambda_{i}-\log(1+\lambda_{i})\big)
=\frac14\|\Sigma\|_{F}^{2}\big(1+O(\|\Sigma\|_{\mathrm{op}})\big)
\]
and $W_2\big(\mathcal N(0,I_d+\Sigma),\mathcal N(0,I_d)\big)=\|(I_d+\Sigma)^{1/2}-I_d\|_F=\frac12\|\Sigma\|_F\big(1+O(\|\Sigma\|_{\mathrm{op}})\big)$. The bounds of Theorem~\ref{thm:poincare-KL} hold for any square-integrable Stein kernel, whereas the estimate $\E\|\tau_\mu(X)\|_F^2\leq C_P\mathrm{tr}(\Sigma)$ is guaranteed only for the kernel of \cite[Theorem~2.4]{courtade2019stein}.
\end{remark}

\section{A coupling approach for bounding \texorpdfstring{$W_p(\mu * \g,\g)$}{Wp(mu*g,g)}}

The proof proceeds in three steps. First, we bound the cost of replacing a symmetric two-component Gaussian mixture by the Gaussian centred at its midpoint. Second, for a centred finite mixture, we choose a perfect matching whose paired points have large total separation; replacing each pair by its midpoint contracts the empirical $2p$-moment. Iterating this midpoint replacement drives the mixture to the standard Gaussian, and summing the one-step costs gives the desired finite-mixture bound. Finally, a centred empirical approximation extends the result to general centred measures with finite $2p$-moment.

\subsection{{A one-dimensional two-component Gaussian mixture bound}}\label{sec:wass_1d}

In one dimension the monotone rearrangement is optimal for the cost $|x-y|^p$, $p\geq 1$, so that
\begin{equation}\label{eq:wass_1d}
W_p(P,Q)
=
\left( \int_{0}^{1} \bigl|F_P^{-1}(u) - F_Q^{-1}(u)\bigr|^p  du \right)^{1/p},
\end{equation}
where $F_P^{-1}$ and $F_Q^{-1}$ are the quantile functions of $P,Q\in\mathcal{P}_{p}(\mathbb{R})$ \citep[see][Theorem~2.18]{villani2003topics}.

\begin{lemma}\label{lem:1Dtwocomp}
Let $P_0 = \mathcal{N}(0,1)$ and $P_\delta = \tfrac{1}{2}\mathcal{N}(\delta,1) + \tfrac{1}{2}\mathcal{N}(-\delta,1)$ for $\delta \geq 0$.
Then, for any $p \geq 1$,
\[
W_p(P_0, P_\delta) \leq K_p \delta^2,
\]
where $K_p = \max\left\{1, \frac{C_p}{2} + \frac{1}{3}\right\}$ and
$C_p = \left(\mathbb{E}_{\xi \sim \mathcal{N}(0,1)}|\xi|^p\right)^{1/p}$.
\end{lemma}

\begin{proof}
We consider the path of measures connecting $P_0$ and $P_\delta$. For $t \in [0, \delta]$, define
\[
P_t = \tfrac{1}{2}\mathcal{N}(t,1) + \tfrac{1}{2}\mathcal{N}(-t,1).
\]
Write $F_t$ and $f_t$ for the CDF and density of $P_t$. {Since $f_t$ is smooth and strictly positive, the quantile function $q_t(u):=F_t^{-1}(u)$ is differentiable in $t$ for each $u\in(0,1)$. Implicit differentiation of $F_t(q_t(u))=u$ gives}
\[
{
\partial_t q_t(u)
=
-\frac{\partial_t F_t(q_t(u))}{f_t(q_t(u))}.
}
\]
{A direct calculation gives}
\[
{
-\frac{\partial_tF_t(x)}{f_t(x)}
=
\tanh(tx).
}
\]
Therefore, by the quantile representation \eqref{eq:wass_1d}, the identity
$q_\delta-q_0=\int_0^\delta\partial_tq_t\,dt$ and Minkowski's integral inequality,
and the substitution $x=q_t(u)$ (under which $u\sim\mathrm{Unif}(0,1)$ gives $x\sim P_t$),
\[
{
W_p(P_0,P_\delta)
=
\|q_\delta-q_0\|_{L^p(0,1)}
\le
\int_0^\delta \|\partial_tq_t\|_{L^p(0,1)}\,dt
=
\int_0^\delta
\left(\int |\tanh(tx)|^p\,P_t(dx)\right)^{1/p}dt .
}
\]
Using the inequality $|\tanh(u)| \leq \min(|u|, 1)$, we can bound, for $\xi \sim \mathcal{N}(0,1)$,
\begin{align*}
\left(\int |\tanh(tx)|^p\,P_t(dx)\right)^{1/p}
&\leq \left( \mathbb{E}\big[ \min(t|\xi + t|, 1)^p \big] \right)^{1/p} \\
&\leq \min\left(t(C_p + t), 1\right),
\end{align*}
where the last inequality uses the triangle inequality and $C_p = (\mathbb{E}|\xi|^p)^{1/p}$.
Thus
\[
W_p(P_0, P_\delta) \leq \int_0^\delta \min\left(t(C_p + t), 1\right)dt.
\]
If $\delta\leq 1$, the integral is at most $\int_0^\delta t(C_p+t)dt\leq(\frac{C_p}{2}+\frac13)\delta^2$; if $\delta\geq 1$, it is at most $\delta\leq\delta^2$.
\end{proof}

\subsection{Bounding \texorpdfstring{$p$}{p}-Wasserstein distance for centred Gaussian mixtures in \texorpdfstring{$d$}{d}-dimensions}

Fix $n\in\mathbb{N}$, and points $x_1,\dots,x_{2n}\in\mathbb{R}^d$ satisfying
\begin{equation}\label{eq:mean_zero}
\sum_{i=1}^{2n}x_i=0.
\end{equation}
Define the Gaussian mixture
\[
\mu_x := \frac{1}{2n}\sum_{i=1}^{2n}\mathcal{N}(x_i, I_d).
\]
Given a perfect matching $M$ of $\{1,\dots,2n\}$ into $n$ disjoint pairs $(i,j)$, define midpoints
\[
m_{ij}:=\frac{x_i+x_j}{2}\qquad ((i,j)\in M),
\]
and form a second set of points $y_1,\dots,y_{2n}$ by setting $y_i=y_j=m_{ij}$ for each $(i,j)\in M$
(so each midpoint is repeated twice). Define
\[
\mu_y := \frac{1}{2n}\sum_{i=1}^{2n}\mathcal{N}(y_i, I_d)
=
\frac{1}{n}\sum_{(i,j)\in M}\mathcal{N}\left(\frac{x_i+x_j}{2}, I_d\right).
\]
For $q\geq 1$, write the empirical $q$-moment
\[
\Phi_q(x):=\frac{1}{2n}\sum_{i=1}^{2n}\|x_i\|^q,
\]
and the energy of the perfect matching $M$ by
\[
S_{2p}(M):=\sum_{(i,j)\in M}\|x_i-x_j\|^{2p}.
\]
\paragraph{{A large-separation perfect matching.}}
\begin{theorem}\label{thm:large_energy}
Assume \eqref{eq:mean_zero} and $p\geq 1$. Then there exists a perfect matching $M$ such that
\begin{equation}\label{eq:large_energy}
S_{2p}(M)\ge\frac{n}{2n-1}\sum_{i=1}^{2n}\|x_i\|^{2p}
\ge\frac{1}{2}\sum_{i=1}^{2n}\|x_i\|^{2p}.
\end{equation}
\end{theorem}

\begin{proof}
Fix $i$ and let $V$ be uniform on $\{x_1,\dots,x_{2n}\}$. Since $v\mapsto \|x_i-v\|^{2p}$ is convex,
Jensen's inequality gives
\[
\frac{1}{2n}\sum_{j=1}^{2n}\|x_i-x_j\|^{2p}=\mathbb{E}\|x_i-V\|^{2p}
\geq \|x_i-\mathbb{E}V\|^{2p}=\|x_i\|^{2p},
\]
using $\mathbb{E}V=0$ from \eqref{eq:mean_zero}. Summing over $i$ yields
\[
\sum_{i\neq j}\|x_i-x_j\|^{2p}\ge2n\sum_{i=1}^{2n}\|x_i\|^{2p}.
\]
Now let $M$ be a uniformly random perfect matching. For each $i\neq j$, the unordered pair $\{i,j\}$
appears in $M$ with probability $1/(2n-1)$. Hence
\[
\mathbb{E}S_{2p}(M)=\frac{1}{2(2n-1)}\sum_{i\neq j}\|x_i-x_j\|^{2p}
\geq \frac{n}{2n-1}\sum_{i=1}^{2n}\|x_i\|^{2p},
\]
and therefore some matching attains at least this value, proving \eqref{eq:large_energy}.
\end{proof}

\paragraph{Midpoint contraction of the large energy perfect matching.}

We use the standard uniform convexity inequality to prove contraction of the $2p$-th moment:
for $q\geq 2$ and all $u,v\in\mathbb{R}^d$,
\begin{equation}\label{eq:clarkson}
\left\|\frac{u+v}{2}\right\|^{q}+\left\|\frac{u-v}{2}\right\|^{q}
\le\frac12\left(\|u\|^{q}+\|v\|^{q}\right).
\end{equation}

\begin{lemma}\label{lem:moment_contraction}
Fix $p\geq 1$ and set $q:=2p\geq 2$. Let $M$ be a perfect matching and form $y_1,\dots,y_{2n}$ by midpoint
replacement along $M$. If
\begin{equation}\label{eq:energy_half}
S_{2p}(M)\ge\frac12\sum_{i=1}^{2n}\|x_i\|^{2p},
\end{equation}
then
\begin{equation}\label{eq:moment_contract}
\Phi_{2p}(y)\le\left(1-2^{-2p}\right)\Phi_{2p}(x).
\end{equation}
\end{lemma}

\begin{proof}
Apply \eqref{eq:clarkson} with $q=2p$, $u=x_i$, $v=x_j$, and note that $m_{ij}=(x_i+x_j)/2$. We get
\[
\|m_{ij}\|^{2p}+\frac{1}{2^{2p}}\|x_i-x_j\|^{2p}\leq \frac12\left(\|x_i\|^{2p}+\|x_j\|^{2p}\right).
\]
Summing over $(i,j)\in M$ yields
\[
\sum_{(i,j)\in M}\|m_{ij}\|^{2p}+\frac{1}{2^{2p}}S_{2p}(M)\leq \frac12\sum_{i=1}^{2n}\|x_i\|^{2p}.
\]
Using \eqref{eq:energy_half} gives
\[
\sum_{(i,j)\in M}\|m_{ij}\|^{2p}
\le
\left(\frac12-\frac{1}{2^{2p+1}}\right)\sum_{i=1}^{2n}\|x_i\|^{2p},
\]
then because each midpoint is repeated twice among $y_1,\dots,y_{2n}$ we have
\[
\Phi_{2p}(y)=\frac{1}{2n}\sum_{i=1}^{2n}\|y_i\|^{2p}
=\frac{1}{n}\sum_{(i,j)\in M}\|m_{ij}\|^{2p}
\le
\left(1-2^{-2p}\right)\Phi_{2p}(x),
\]
as required.
\end{proof}

\paragraph{One-step coupling and chaining.}

\begin{proposition}\label{prop:one_step_mixture}
Let $M$ be a perfect matching on $\{1,\dots,2n\}$ and let $y$ be obtained from $x$ by midpoint replacement along $M$. Then for every $p\geq 1$,
\begin{equation}\label{eq:one_step_quadratic}
W_p(\mu_x,\mu_y)
\leq
\frac{K_p}{4}
\left(\frac1n\sum_{(i,j)\in M}\|x_i-x_j\|^{2p}\right)^{1/p}
\leq
K_p\Phi_{2p}(x)^{1/p}.
\end{equation}
\end{proposition}

\begin{proof}
Fix $(i,j)\in M$, let $\nu_{ij}:=\frac12\mathcal N(x_i,I_d)+\frac12\mathcal N(x_j,I_d)$, $u:=(x_i-x_j)/2$ and $e:=u/\|u\|$ (if $u=0$, then $\nu_{ij}=\mathcal N(\frac{x_i+x_j}{2},I_d)$ and there is nothing to prove). By translation invariance of $W_p$, after subtracting $(x_i+x_j)/2$ it suffices to couple $\frac12\mathcal N(u,I_d)+\frac12\mathcal N(-u,I_d)$ with $\mathcal N(0,I_d)$. Writing $z=se+w$ with $w\perp e$, both laws have independent $\mathcal N(0,I_{d-1})$ components in $e^\perp$, which we take to be equal; the $e$-components are $\frac12\mathcal N(\|u\|,1)+\frac12\mathcal N(-\|u\|,1)$ and $\mathcal N(0,1)$. Lemma~\ref{lem:1Dtwocomp} and translating back therefore yield a coupling $(X_{ij},Y_{ij})$ of $\nu_{ij}$ and $\mathcal N(\frac{x_i+x_j}{2},I_d)$ with $\|X_{ij}-Y_{ij}\|_{L^p}\leq K_p\|x_i-x_j\|^2/4$. Since $\mu_x=\frac1n\sum_{(i,j)\in M}\nu_{ij}$ and $\mu_y=\frac1n\sum_{(i,j)\in M}\mathcal N(\frac{x_i+x_j}{2},I_d)$, choosing a pair $(I,J)$ uniformly from $M$ and then $(X_{IJ},Y_{IJ})$ gives a coupling of $\mu_x$ and $\mu_y$, which proves the first inequality. The second follows from $\|x_i-x_j\|^{2p}\leq 2^{2p-1}(\|x_i\|^{2p}+\|x_j\|^{2p})$, summed over $M$.
\end{proof}

Let $x^{(0)}:=x$ and, given $x^{(t)}$, let $x^{(t+1)}$ be obtained by midpoint replacement along a perfect matching $M_t$ with $S_{2p}(M_t)\geq\frac12\sum_i\|x^{(t)}_i\|^{2p}$, which exists by Theorem~\ref{thm:large_energy}. Midpoint replacement preserves $\sum_i x^{(t)}_i=0$. Let $\mu_t := \frac{1}{2n}\sum_{i}\mathcal{N}(x^{(t)}_i, I_d)$ and $\Phi_t := \Phi_{2p}(x^{(t)})$.

\begin{theorem}\label{thm:chaining}
Assume \eqref{eq:mean_zero} and fix $p\geq 1$. Then
\begin{equation}\label{eq:final_bound_corrected}
W_p\left(\frac{1}{2n}\sum_{i=1}^{2n}\mathcal{N}(x_i, I_d),\g\right)
\le
\frac{K_p}{1-\left(1-2^{-2p}\right)^{1/p}}
\left(\frac{1}{2n}\sum_{i=1}^{2n}\|x_i\|^{2p}\right)^{1/p}.
\end{equation}
\end{theorem}

\begin{proof}
By Lemma~\ref{lem:moment_contraction}, $\Phi_t\leq(1-2^{-2p})^t\Phi_0$, so Proposition~\ref{prop:one_step_mixture} and the triangle inequality give
\[
W_p(\mu_0,\mu_T)\leq\sum_{t=0}^{T-1}K_p\Phi_t^{1/p}\leq K_p\Phi_0^{1/p}\sum_{t=0}^{\infty}(1-2^{-2p})^{t/p}=\frac{K_p\Phi_0^{1/p}}{1-(1-2^{-2p})^{1/p}}.
\]
Coupling $\mu_T$ and $\g$ through $x^{(T)}_I+G$ and $G$, with $I$ uniform and $G\sim\g$, gives $W_p(\mu_T,\g)\leq\Phi_T^{1/(2p)}\to0$, and letting $T\to\infty$ proves the claim.
\end{proof}

\convGeneral*

\begin{proof}
Let $X_1,X_2,\dots$ be i.i.d.\ with law $\mu$, let $\bar X_{2n}$ be their empirical mean over the first $2n$ terms, and let $\mu^{(n)}$ be the empirical measure of the centred points $x^{(n)}_i:=X_i-\bar X_{2n}$, $1\leq i\leq 2n$. Theorem~\ref{thm:chaining} applies to $\mu^{(n)}*\g$. Since adding the same Gaussian noise to both marginals of a coupling shows $W_p(\alpha*\g,\beta*\g)\leq W_p(\alpha,\beta)$,
\[
W_p(\mu*\g,\g)
\le
W_p(\mu,\mu^{(n)})+\frac{K_p}{1-\left(1-2^{-2p}\right)^{1/p}}\left(\frac{1}{2n}\sum_{i=1}^{2n}\|x^{(n)}_i\|^{2p}\right)^{1/p}.
\]
By the strong law of large numbers, $\bar X_{2n}\to0$ almost surely; moreover, the uncentred empirical measures converge weakly to $\mu$ and their $p$-th moments converge almost surely, so they converge to $\mu$ in $W_p$, and since $W_p(\mu^{(n)},\frac{1}{2n}\sum_{i\leq 2n}\delta_{X_i})\leq\|\bar X_{2n}\|$, also $W_p(\mu^{(n)},\mu)\to 0$. Finally, $\frac{1}{2n}\sum_i\|x^{(n)}_i\|^{2p}\to\mathbb{E}\|X\|^{2p}$ almost surely, and letting $n\to\infty$ gives \eqref{eq:general_mu_convolution_bound_main}.
\end{proof}

\section{A stochastic-localization proof of a Gaussian convolution bound for \texorpdfstring{$p\geq 2$}{p >= 2}}
\label{app:stochastic-localization-wp}

In this appendix, we give an alternative proof of a Gaussian convolution inequality for $W_p$, $p\geq 2$, based on the planted, or Bayesian, form of stochastic localization \citep[see][]{ChenEldanSL,klartag2023spectral,el2022information}. This argument does not replace Theorem~\ref{thm:conv_general}, since the latter also covers $1\leq p<2$, which is important for the $W_1$ applications. For $p\geq 2$, however, the stochastic-localization argument gives a sharper dimension-free estimate with a simple constant.

We use two auxiliary results. The first is a standard consequence of the Benamou--Brenier formulation of optimal transport (see \citealp{benamou2000computational} and \citealp[Theorem~8.3.1]{ambrosio2008gradient}).

\begin{lemma}
\label{lem:ce-length}
Let $p>1$ and let $(\rho_t)_{0\leq t\leq 1}$ be a narrowly continuous curve in $\mathcal P_p(\mathbb{R}^d)$ solving the continuity equation $\partial_t\rho_t+\nabla\cdot(\rho_tv_t)=0$ in the distributional sense, for a Borel vector field $v$ with $\int_0^1\|v_t\|_{L^p(\rho_t)}dt<\infty$. Then
\[
W_p(\rho_1,\rho_0)
\leq
\int_0^1
\left(
\int_{\mathbb{R}^d}\|v_t(x)\|^p\rho_t(dx)
\right)^{1/p}dt.
\]
\end{lemma}

\begin{lemma}
\label{lem:vector-martingale-estimate}
Let $p\geq2$, let $W$ be a standard Brownian motion in $\mathbb{R}^d$, and let $H$ be a predictable $\mathbb{R}^{d\times d}$-valued process with $\int_0^T\|H_s\|_{\mathrm{HS}}^2ds<\infty$ a.s. Then $M_T:=\int_0^T H_s\,dW_s$ satisfies
\[
\|M_T\|_{L^p}^2
\leq
(p-1)\int_0^T\|H_s\|_{L^p(\mathrm{HS})}^2ds,
\qquad
\|H_s\|_{L^p(\mathrm{HS})}
:=
\left(\mathbb{E}\|H_s\|_{\mathrm{HS}}^p\right)^{1/p}.
\]
\end{lemma}

\begin{proof}
For $p=2$ this is It\^o's isometry. For $p>2$ and bounded $H$, let $U(t):=\mathbb{E}\|M_t\|^p$. Since $D^2(\|\cdot\|^p)(m)\preceq p(p-1)\|m\|^{p-2}I_d$, It\^o's formula and H\"older's inequality give
\[
U'(t)\leq\frac{p(p-1)}{2}\mathbb{E}\left[\|M_t\|^{p-2}\|H_t\|_{\mathrm{HS}}^2\right]\leq\frac{p(p-1)}{2}U(t)^{(p-2)/p}\|H_t\|_{L^p(\mathrm{HS})}^2
\]
for a.e.\ $t$. Hence $V_\varepsilon(t):=(U(t)+\varepsilon)^{2/p}$ satisfies $V_\varepsilon'(t)\leq(p-1)\|H_t\|_{L^p(\mathrm{HS})}^2$; integrating from $U(0)=0$ and letting $\varepsilon\downarrow0$ proves the bound. For general $H$, we may assume that the right-hand side is finite, since otherwise there is nothing to prove; then H\"older's inequality ($p\geq 2$) and Tonelli's theorem give $\mathbb{E}\int_0^T\|H_s\|_{\mathrm{HS}}^2ds\leq\int_0^T\|H_s\|_{L^p(\mathrm{HS})}^2ds<\infty$. Apply the bound to $H\mathbf 1_{\{\|H\|_{\mathrm{HS}}\leq m\}}$, whose stochastic integrals converge to $M_T$ in $L^2$ by It\^o's isometry, and pass to the limit along an almost surely convergent subsequence using Fatou's lemma and monotone convergence.
\end{proof}

We now prove the convolution bound via a stochastic localization process known as the Bayesian, or planted, version of Eldan's stochastic localization \citep[see][]{eldan2013thin,klartag2023spectral,el2022information,lee2017eldans}.

\begin{theorem}
\label{thm:stochastic-localization-wp}
Let $d\geq1$, $p\geq2$, and let $\mu\in\mathcal P(\mathbb{R}^d)$ be centred with
\[
\int_{\mathbb{R}^d}\|x\|^{2p}\mu(dx)<\infty.
\]
Let $\g=\mathcal N(0,I_d)$, then
\[
W_p(\mu*\g,\g)
\leq
\frac{\sqrt{p-1}}{2}
\left(
\int_{\mathbb{R}^d}\|x\|^{2p}\mu(dx)
\right)^{1/p}.
\]
\end{theorem}

\begin{proof}
Let $X\sim\mu$ and first assume that $X$ is bounded. The general case is obtained at the end by truncation. Let $(W_s)_{s\geq0}$ be a standard Brownian motion in $\mathbb{R}^d$, independent of $X$, and define the Gaussian observation process
\[
S_s=sX+W_s,
\qquad s\geq0.
\]
For $s\geq0$, let $\mathcal F_s:=\sigma(S_u:0\leq u\leq s)$, then the posterior law of $X$ given $\mathcal F_s$ is
\[
\mu_s(dx) := \mathbb{P}(X\in dx\mid\mathcal F_s) = \frac{ \exp\left(\langle S_s,x\rangle-\frac{s}{2}\|x\|^2\right)}{\displaystyle \int_{\mathbb{R}^d} \exp\left(\langle S_s,z\rangle-\frac{s}{2}\|z\|^2\right)\mu(dz)}\mu(dx).
\]
We note that $S_{s}$ is a sufficient statistic for the observation path up to time $s$. Define the posterior mean and covariance by
\[
a_s
:=
\int_{\mathbb{R}^d}x\mu_s(dx)
=
\mathbb{E}[X\mid\mathcal F_s],
\qquad
A_s
:=
\int_{\mathbb{R}^d}(x-a_s)(x-a_s)^\top\mu_s(dx).
\]
Since $\mu$ is centred, $a_0=0$. 

We next derive the stochastic differential equation for $a_s$. For bounded measurable $f$, set
\[
L_s(x) := \exp\left(\langle S_s,x\rangle-\frac{s}{2}\|x\|^2 \right), \qquad Q_s(f) := \int_{\mathbb{R}^d}f(x)L_s(x)\mu(dx),
\]
then $\mu_s(f)=Q_s(f)/Q_s(1)$. Since $[S]_s=sI_d$, It\^o's formula gives
\[
dL_s(x)=L_s(x)\langle x,dS_s\rangle.
\]
Consequently,
\[
dQ_s(f)=Q_s(fx)\cdot dS_s,
\qquad
dQ_s(1)=Q_s(x)\cdot dS_s,
\]
where $Q_s(fx)$ denotes the vector with coordinates $Q_s(fx_i)$. The quotient rule yields
\[
d\mu_s(f)
=
\big(\mu_s(fx)-\mu_s(f)a_s\big)
\cdot(dS_s-a_sds).
\]

The innovation process $B_s:=S_s-\int_0^s a_u\,du$ is an $(\mathcal F_s)$-Brownian motion: for $u\geq s$, $B_u-B_s=\int_s^u(X-a_r)dr+(W_u-W_s)$ has zero conditional mean given $\mathcal F_s$, because $a_r=\mathbb{E}[X\mid\mathcal F_r]$ is a martingale and $\mathcal F_s\subseteq\sigma(X,W_r:r\leq s)$; since $[B]_s=sI_d$, L\'evy's characterization applies.

Therefore
\[
d\mu_s(f)=\big(\mu_s(fx)-\mu_s(f)a_s\big)\cdot dB_s.
\]
Taking $f(x)=x_i$, $i=1,\ldots,d$, yields the localization martingale identity
\[
da_s=A_sdB_s,\qquad a_0=0.
\]

We now estimate $a_s$; via Lemma~\ref{lem:vector-martingale-estimate} we have
\[
\|a_s\|_{L^p}^2
\leq
(p-1)\int_0^s \|A_u\|_{L^p(\mathrm{HS})}^2du .
\]
Since $A_u\succeq0$,
\[
\|A_u\|_{\mathrm{HS}}
\leq
\operatorname{Tr}(A_u)
=
\mathbb E[\|X-a_u\|^2\mid\mathcal F_u]
=
\mathbb E[\|X\|^2\mid\mathcal F_u]-\|a_u\|^2
\le
\mathbb E[\|X\|^2\mid\mathcal F_u].
\]
By conditional Jensen's inequality,
\[
\|A_u\|_{L^p(\mathrm{HS})}
\le
\left(\mathbb E\left[\mathbb E(\|X\|^2\mid\mathcal F_u)^p\right]\right)^{1/p}
\le
\left(\mathbb E\|X\|^{2p}\right)^{1/p}.
\]
Consequently,
\[
\|a_s\|_{L^p}
\le
\sqrt{p-1}\sqrt{s}
\left(\mathbb E\|X\|^{2p}\right)^{1/p}.
\]

We now connect this posterior-mean estimate to Gaussian convolution. Let $Z\sim \g$ be independent of $X$, and define
\[
Y_r:=Z+rX,\qquad \rho_r:=\mathcal L(Y_r),\qquad 0\leq r\leq 1,
\]
then $\rho_0=\g$ and $\rho_1=\mu*\g$. Now define
\[
v_r(y):=\mathbb E[X\mid Y_r=y],
\]
and note that for every $\varphi\in C_c^\infty(\mathbb R^d)$, boundedness of $X$ justifies differentiating under the expectation and gives
\[
\frac{d}{dr}\int_{\mathbb R^d}\varphi(y)\rho_r(dy)
=
\frac{d}{dr}\mathbb E\varphi(Z+rX)
=
\mathbb E\langle X,\nabla\varphi(Z+rX)\rangle
=
\int_{\mathbb R^d}\langle v_r(y),\nabla\varphi(y)\rangle\rho_r(dy).
\]
Thus $(\rho_r,v_r)$ solves the continuity equation and Lemma~\ref{lem:ce-length} gives
\[
W_p(\mu*\g,\g)
\le
\int_0^1
\left(\mathbb E\|\mathbb E[X\mid Y_r]\|^p\right)^{1/p}
dr.
\]
For $r>0$, set $s=r^2$, then since
\[
\frac{S_{r^2}}{r}
=
rX+\frac{W_{r^2}}{r},
\qquad
\frac{W_{r^2}}{r}\sim N(0,I_d),
\]
and $W_{r^2}/r$ is independent of $X$, the joint laws of $(X,S_{r^2}/r)$ and $(X,Y_r)$ coincide. By the sufficiency of $S_{r^2}$ for the observation path up to time $r^2$,
\[
\left(\mathbb E\|\mathbb E[X\mid Y_r]\|^p\right)^{1/p}
=
\|a_{r^2}\|_{L^p}.
\]
Combining the preceding estimates,
\[
W_p(\mu*\g,\g)
\le
\int_0^1 \|a_{r^2}\|_{L^p}dr
\le
\int_0^1
\sqrt{p-1}r
\left(\mathbb E\|X\|^{2p}\right)^{1/p}
dr,
\]
and hence
\[
W_p(\mu*\g,\g)
\le
\frac{\sqrt{p-1}}{2}
\left(\mathbb E\|X\|^{2p}\right)^{1/p}.
\]
This proves the theorem when $X$ is bounded.

We now remove the boundedness assumption. Let
\[
X_n:=X\mathbf 1_{\{\|X\|\leq n\}}
-\mathbb E\left[X\mathbf 1_{\{\|X\|\leq n\}}\right],
\qquad
\mu_n:=\mathcal L(X_n).
\]
Then $X_n$ is bounded and centred, and $X_n\to X$ in $L^{2p}$. The bounded case gives
\[
W_p(\mu_n*\g,\g)
\le
\frac{\sqrt{p-1}}{2}
\left(\mathbb E\|X_n\|^{2p}\right)^{1/p}.
\]
Coupling $X_n+Z$ and $X+Z$ with the same Gaussian $Z$ gives
\[
W_p(\mu_n*\g,\mu*\g)
\le
\left(\mathbb E\|X_n-X\|^p\right)^{1/p}
\longrightarrow 0.
\]
Also, since $X_n\to X$ in $L^{2p}$,
\[
\left(\mathbb E\|X_n\|^{2p}\right)^{1/p}
\longrightarrow
\left(\mathbb E\|X\|^{2p}\right)^{1/p},
\]
and passing to the limit proves the result.
\end{proof}

\section{Stochastic-gradient estimator for the xerophthalmia example}\label{app:xerop}
We give the estimator used in Section~\ref{sec:xerop}. For each child $t$, let $f_t(\cdot;\theta)$ be as in Section~\ref{sec:xerop}. Exact draws from $p(b\mid\theta,y)$ could be obtained by vectorised adaptive rejection sampling \citep{gilks1992adaptive}, but importance sampling against a tangent-line proposal gives a lower-variance estimator at comparable cost. Concretely, write $\mu_{t}(\theta)$ for the conditional mode of $f_{t}(\cdot;\theta)$ (computed in five Newton iterations) and $\sigma_{t}(\theta)=(-f_{t}''(\mu_{t}))^{-1/2}$ for the curvature scale at the mode. The estimator combines the following three ingredients.
\begin{enumerate}[(i)]
\item a two-tangent piecewise-exponential proposal $q_{t}(b;\theta)\propto \exp{\left(L_{t}(b;\theta)\right)}$ tangent to $f_{t}$ at $\mu_{t}\pm\sigma_{t}$, with closed-form normalising constant $Z_{q,t}(\theta)$;
\item $n_{\mathrm{re}}$ stratified samples $b_{t}^{(1)},\dots,b_{t}^{(n_{\mathrm{re}})}$ from $q_{t}(\cdot;\theta)$ via stratified inverse-CDF with antithetic Latin-hypercube jitter, $u_{n_{\mathrm{re}}-1-i}=1-u_{i}$;
\item a mode-centred $K$-point Gauss--Hermite quadrature for the true normaliser $Z_{t}(\theta)=\int_{\R}e^{f_{t}(b;\theta)}db$,
\[
\widehat{Z}_{t,\mathrm{GH}}(\theta)=\sigma_{t}\sum_{k=1}^{K}\omega_{k}\exp\left(\xi_{k}^{2}+f_{t}(\mu_{t}+\sigma_{t}\xi_{k};\theta)\right),
\]
with $K=8$ Hermite nodes $(\xi_{k},\omega_{k})$; compared with a $128$-node rule, the relative error of $\widehat{Z}_{t,\mathrm{GH}}$ at the reported posterior means is below $6\times10^{-4}$ across all children. This rescales each child's gradient contribution by a factor $1+\mathcal{O}(10^{-4})$, which is negligible relative to the Monte Carlo error in Table~\ref{tab:xerop-summary}.
\end{enumerate}
Combining (i)--(iii) yields the following importance-sampling estimator of the stochastic gradient $\nabla V(\theta)=-\nabla\log\pi(\theta)$:
\begin{align}
\label{eq:xerop-grad}
\mathcal{G}(\theta,\omega)
&=
-\nabla\log p(\theta)
-\sum_{t=1}^{\ell}\frac{1}{n_{\mathrm{re}}}\sum_{i=1}^{n_{\mathrm{re}}}w_{t}^{(i)}(\theta)\nabla_{\theta}\log p\left(y_{t,\cdot},b_{t}^{(i)}\mid\theta\right),\\
\nonumber w_{t}^{(i)}(\theta)&=\frac{Z_{q,t}(\theta)}{\widehat{Z}_{t,\mathrm{GH}}(\theta)}\exp\left(f_{t}(b_{t}^{(i)};\theta)-L_{t}(b_{t}^{(i)};\theta)\right).
\end{align}
With exact conditional sampling and exact normalising integrals, Fisher's identity would make this an unbiased estimator of $\nabla V(\theta)$; the finite Gauss--Hermite normaliser makes it only approximately unbiased. Stratification with antithetic pairs gives roughly a further two-fold variance reduction in the $\beta$-coordinates over plain Monte Carlo, and the deterministic normaliser removes the $\mathcal{O}(1/n_{\mathrm{re}})$ variance that estimating $Z_{t}$ stochastically would add.

\end{appendices}

\vskip 0.2in
\bibliography{sample}

@article{benamou2000computational,
  title={A computational fluid mechanics solution to the Monge-Kantorovich mass transfer problem},
  author={Benamou, Jean-David and Brenier, Yann},
  journal={Numerische Mathematik},
  volume={84},
  number={3},
  pages={375--393},
  year={2000},
  publisher={Springer}
}

@article{leimkuhler2013rational,
  title={Rational construction of stochastic numerical methods for molecular sampling},
  author={Leimkuhler, Benedict and Matthews, Charles},
  journal={Applied Mathematics Research eXpress},
  volume={2013},
  number={1},
  pages={34--56},
  year={2013},
  publisher={Oxford University Press}
}

@article{monmarche2021high,
  title={High-dimensional {MCMC} with a standard splitting scheme for the underdamped {L}angevin diffusion.},
  author={Monmarch{\'e}, Pierre},
  journal={Electronic Journal of Statistics},
  volume={15},
  number={2},
  pages={4117--4166},
  year={2021},
  publisher={Institute of Mathematical Statistics and Bernoulli Society}
}

@inproceedings{cheng2018underdamped,
  title={Underdamped {L}angevin {MCMC}: A non-asymptotic analysis},
  author={Cheng, Xiang and Chatterji, Niladri S and Bartlett, Peter L and Jordan, Michael I},
  booktitle={Conference on learning theory},
  pages={300--323},
  year={2018},
  organization={PMLR}
}

@article{dalalyan2020sampling,
  title={On sampling from a log-concave density using kinetic {L}angevin diffusions},
  author={Dalalyan, Arnak S and Riou-Durand, Lionel},
  journal={Bernoulli},
  volume={26},
  number={3},
  pages={1956--1988},
  year={2020},
  publisher={Bernoulli Society for Mathematical Statistics and Probability}
}

@article{sanz2021wasserstein,
  title={Wasserstein distance estimates for the distributions of numerical approximations to ergodic stochastic differential equations.},
  author={Sanz-Serna, Jesus Maria and Zygalakis, Konstantinos C},
  journal={J. Mach. Learn. Res.},
  volume={22},
  number={242},
  pages={1--37},
  year={2021}
}

@article{alamo2016technique,
  title={A technique for studying strong and weak local errors of splitting stochastic integrators},
  author={Alamo, Alfonso and Sanz-Serna, Jes{\'u}s Mar{\'\i}a},
  journal={SIAM Journal on Numerical Analysis},
  volume={54},
  number={6},
  pages={3239--3257},
  year={2016},
  publisher={SIAM}
}

@article{eberle2019couplings,
  title={Couplings and quantitative contraction rates for {L}angevin dynamics},
  author={Eberle, Andreas and Guillin, Arnaud and Zimmer, Raphael},
  journal={The Annals of Probability},
  volume={47},
  number={4},
  pages={1982--2010},
  year={2019},
  publisher={Institute of Mathematical Statistics}
}

@article{dalalyan2017theoretical,
  title={Theoretical guarantees for approximate sampling from smooth and log-concave densities},
  author={Dalalyan, Arnak S},
  journal={Journal of the Royal Statistical Society: Series B (Statistical Methodology)},
  volume={79},
  number={3},
  pages={651--676},
  year={2017},
  publisher={Wiley Online Library}
}

@article{durmus2017nonasymptotic,
  title={Nonasymptotic convergence analysis for the unadjusted {L}angevin algorithm},
  author={Durmus, Alain and Moulines, Eric},
  journal={The Annals of Applied Probability},
  volume={27},
  number={3},
  pages={1551--1587},
  year={2017},
  publisher={Institute of Mathematical Statistics}
}

@book{villani2009optimal,
  title={Optimal transport: old and new},
  author={Villani, C{\'e}dric},
  volume={338},
  year={2009},
  publisher={Springer}
}

@book{pavliotis2014stochastic,
  title={Stochastic processes and applications: diffusion processes, the {F}okker-{P}lanck and {L}angevin equations},
  author={Pavliotis, Grigorios A},
  volume={60},
  year={2014},
  publisher={Springer}
}

@article{schuh2022global,
  title={Global contractivity for {L}angevin dynamics with distribution-dependent forces and uniform in time propagation of chaos},
  author={Schuh, Katharina},
  journal={arXiv preprint arXiv:2206.03082},
  year={2022}
}

@article{vollmer2016exploration,
    AUTHOR = {Vollmer, Sebastian J. and Zygalakis, Konstantinos C. and Teh,
              Yee Whye},
     TITLE = {Exploration of the (non-)asymptotic bias and variance of
              stochastic gradient {L}angevin dynamics},
   JOURNAL = {J. Mach. Learn. Res.},
  FJOURNAL = {Journal of Machine Learning Research (JMLR)},
    VOLUME = {17},
      YEAR = {2016},
     PAGES = {Paper No. 159, 45},
      ISSN = {1532-4435,1533-7928},
   MRCLASS = {60J22 (60J70 62F15 65C05 68T10)},
  MRNUMBER = {3555050},
}

@article{skeel2002impulse,
  title={An impulse integrator for {L}angevin dynamics},
  author={Skeel, Robert D and Izaguirre, Jes{\"u}s A},
  journal={Molecular Physics},
  volume={100},
  number={24},
  pages={3885--3891},
  year={2002},
  publisher={Taylor \& Francis}
}

@article{durmus2019high,
  title={High-dimensional {B}ayesian inference via the unadjusted {L}angevin algorithm},
  author={Durmus, Alain and Moulines, Eric},
  journal={Bernoulli},
  volume={25},
  number={4A},
  pages={2854--2882},
  year={2019},
  publisher={Bernoulli Society for Mathematical Statistics and Probability}
}

@inproceedings{welling2011bayesian,
  title={Bayesian learning via stochastic gradient {L}angevin dynamics},
  author={Welling, Max and Teh, Yee W},
  booktitle={Proceedings of the 28th international conference on machine learning (ICML-11)},
  pages={681--688},
  year={2011}
}

@article{brunger1984stochastic,
  title={Stochastic boundary conditions for molecular dynamics simulations of {ST2} water},
  author={Br{\"u}nger, Axel and Brooks III, Charles L and Karplus, Martin},
  journal={Chemical physics letters},
  volume={105},
  number={5},
  pages={495--500},
  year={1984},
  publisher={Elsevier}
}

@article{robbins1951stochastic,
  title={A stochastic approximation method},
  author={Robbins, Herbert and Monro, Sutton},
  journal={The annals of mathematical statistics},
  pages={400--407},
  year={1951},
  publisher={JSTOR}
}

@article{leimkuhler2023contractionb,
    AUTHOR = {Leimkuhler, Benedict and Paulin, Daniel and Whalley, Peter A.},
     TITLE = {Contraction rate estimates of stochastic gradient kinetic
              {L}angevin integrators},
   JOURNAL = {ESAIM Math. Model. Numer. Anal.},
  FJOURNAL = {ESAIM. Mathematical Modelling and Numerical Analysis},
    VOLUME = {58},
      YEAR = {2024},
    NUMBER = {6},
     PAGES = {2255--2286},
      ISSN = {2822-7840,2804-7214},
   MRCLASS = {65C05 (65C30 65C40)},
  MRNUMBER = {4834894},
       DOI = {10.1051/m2an/2024038},
       URL = {https://doi.org/10.1051/m2an/2024038},
}

@phdthesis{BUBthesis,
  author  = "Alfonso Álamo Zapatero",
  title   = "Word Series for the Numerical Integration of Stochastic Differential Equations",
  school  = "Universidad de Valladolid",
  year    = "2017"
}

@book{villani2003topics,
  title={Topics in optimal transportation},
  author={Villani, C{\'e}dric},
  volume={58},
  year={2003},
  publisher={American Mathematical Society},
  address={Providence, RI},
  series={Graduate Studies in Mathematics},
  isbn={0-8218-3312-X}
}

@article{dur-maj-mia-2019,
  title={Analysis of Langevin Monte Carlo via convex optimization},
  author={Durmus, Alain and Majewski, Szymon and Miasojedow, B{\l}a{\.z}ej},
  journal={Journal of Machine Learning Research},
  volume={20},
  number={73},
  pages={1--46},
  year={2019}
}

@article {DuEb24,
    AUTHOR = {Durmus, Alain and Eberle, Andreas},
     TITLE = {Asymptotic bias of inexact {M}arkov chain {M}onte {C}arlo
              methods in high dimension},
   JOURNAL = {Ann. Appl. Probab.},
  FJOURNAL = {The Annals of Applied Probability},
    VOLUME = {34},
      YEAR = {2024},
    NUMBER = {4},
     PAGES = {3435--3468},
      ISSN = {1050-5164,2168-8737},
   MRCLASS = {60J05 (65C05)},
  MRNUMBER = {4783022},
MRREVIEWER = {Feng-Rung\ Hu},
       DOI = {10.1214/23-aap2034},
       URL = {https://doi.org/10.1214/23-aap2034},
}

@article{shaw2025random,
  title={Random reshuffling for stochastic gradient {L}angevin dynamics},
  author={Shaw, Luke and Whalley, Peter A},
  journal={arXiv preprint arXiv:2501.16055},
  year={2025}
}

@article{lu2025mean,
  title={Mean square error analysis of stochastic gradient and variance-reduced sampling algorithms},
  author={Lu, Jianfeng and Ye, Xuda and Zhou, Zhennan},
  journal={arXiv preprint arXiv:2511.04413},
  year={2025}
}

@unpublished{che-book-2024+,
author = {Chewi, Sinho},
title = {Log-concave sampling},
note = {Draft available at \url{https://chewisinho.github.io/}},
year = {2025}
}

@article{lu2022explicit,
    AUTHOR = {Cao, Yu and Lu, Jianfeng and Wang, Lihan},
     TITLE = {On explicit {$L^2$}-convergence rate estimate for underdamped
              {L}angevin dynamics},
   JOURNAL = {Arch. Ration. Mech. Anal.},
  FJOURNAL = {Archive for Rational Mechanics and Analysis},
    VOLUME = {247},
      YEAR = {2023},
    NUMBER = {5},
     PAGES = {Paper No. 90, 34},
}

@article{leimkuhler2023contractiona,
author = {Leimkuhler, Benedict J. and Paulin, Daniel and Whalley, Peter A.},
title = {Contraction and {C}onvergence {R}ates for {D}iscretized {K}inetic {L}angevin {D}ynamics},
journal = {SIAM Journal on Numerical Analysis},
volume = {62},
number = {3},
pages = {1226-1258},
year = {2024},
}

@article{SGHMC,
  title={{S}tochastic gradient {H}amiltonian {M}onte {C}arlo},
  author={Chen, Tianqi and Fox, Emily and Guestrin, Carlos},
  journal={Proceedings of the 31st International Conference on Machine Learning},
  volume={32},
  number={2},
  pages={1683--1691},
  year={2014}
}

@article {gouraud2022hmc,
    AUTHOR = {Gouraud, Nicola\"i{} and Le Bris, Pierre and Majka, Adrien and
              Monmarch\'e, Pierre},
     TITLE = {H{MC} and underdamped {L}angevin united in the unadjusted
              convex smooth case},
   JOURNAL = {SIAM/ASA J. Uncertain. Quantif.},
  FJOURNAL = {SIAM/ASA Journal on Uncertainty Quantification},
    VOLUME = {13},
      YEAR = {2025},
    NUMBER = {1},
     PAGES = {278--303},
      ISSN = {2166-2525},
   MRCLASS = {65C40 (60J22 60J60 65C05)},
  MRNUMBER = {4876552},
       DOI = {10.1137/23M1608963},
       URL = {https://doi.org/10.1137/23M1608963},
}

@article {BoOw10,
    AUTHOR = {Bou-Rabee, Nawaf and Owhadi, Houman},
     TITLE = {Long-run accuracy of variational integrators in the stochastic
              context},
   JOURNAL = {SIAM J. Numer. Anal.},
  FJOURNAL = {SIAM Journal on Numerical Analysis},
    VOLUME = {48},
      YEAR = {2010},
    NUMBER = {1},
     PAGES = {278--297},
      ISSN = {0036-1429,1095-7170},
   MRCLASS = {65C30 (60H10 60J05 65C05 65P10)},
  MRNUMBER = {2608370},
MRREVIEWER = {Andreas\ R\"o\ss ler},
       DOI = {10.1137/090758842},
       URL = {https://doi.org/10.1137/090758842},
}

@article{ububu,
  title={Unbiased kinetic {L}angevin {M}onte {C}arlo with inexact gradients},
  author={Chada, Neil K. and Leimkuhler, Benedict and Paulin, Daniel and Whalley, Peter A.},
  journal={arXiv preprint arXiv:2311.05025},
  year={2023}
}

@article{lemast2016,
    AUTHOR = {Leimkuhler, Benedict and Matthews, Charles and Stoltz,
              Gabriel},
     TITLE = {The computation of averages from equilibrium and
              nonequilibrium {L}angevin molecular dynamics},
   JOURNAL = {IMA J. Numer. Anal.},
  FJOURNAL = {IMA Journal of Numerical Analysis},
    VOLUME = {36},
      YEAR = {2016},
    NUMBER = {1},
     PAGES = {13--79},
      ISSN = {0272-4979,1464-3642},
   MRCLASS = {65C30 (60H10 60H35 80A10 82C31)},
  MRNUMBER = {3463433},
MRREVIEWER = {Mikhail\ V.\ Tretyakov},
       DOI = {10.1093/imanum/dru056},
       URL = {https://doi.org/10.1093/imanum/dru056},
}

@article{OBABO,
 title = {Accurate sampling using {L}angevin dynamics},
 author = {Bussi, Giovanni and Parrinello, Michele},
 journal = {Phys. Rev. E},
 volume = {75},
 issue = {5},
 pages = {056707},
 numpages = {7},
 year = {2007},
 month = {May},
 publisher = {American Physical Society},
 doi = {10.1103/PhysRevE.75.056707},
 url = {https://link.aps.org/doi/10.1103/PhysRevE.75.056707}
}

@article{chakreflection,
  title={{R}eflection coupling for unadjusted
generalized {H}amiltonian {M}onte {C}arlo in
the nonconvex stochastic gradient case},
  author={Chak, Martin and Monmarch{\'e}, Pierre },
  journal={arXiv preprint arXiv:2310.18774},
  year={2023}
}

@article{schuh2024convergence,
  title={Convergence of kinetic Langevin samplers for non-convex potentials},
  author={Schuh, Katharina and Whalley, Peter A},
  journal={arXiv preprint arXiv:2405.09992},
  year={2024}
}

@article{courtade2019stein,
  author  = {Courtade, Thomas A. and Fathi, Max and Pananjady, Ashwin},
  title   = {Existence of Stein kernels under a spectral gap, and discrepancy bounds},
  journal = {Annales de l'Institut Henri Poincar\'e, Probabilit\'es et Statistiques},
  volume  = {55},
  number  = {2},
  pages   = {777--790},
  year    = {2019},
  doi     = {10.1214/18-AIHP898},
  eprint  = {1703.07707},
  archivePrefix = {arXiv},
  primaryClass  = {math.PR}
}

@incollection{skeel1999integration,
  title={Integration schemes for molecular dynamics and related applications},
  author={Skeel, Robert D},
  booktitle={The Graduate Student’s Guide to Numerical Analysis’ 98: Lecture Notes from the VIII EPSRC Summer School in Numerical Analysis},
  pages={119--176},
  year={1999},
  publisher={Springer}
}

@article {DalalyanSG,
    AUTHOR = {Dalalyan, Arnak S. and Karagulyan, Avetik},
     TITLE = {User-friendly guarantees for the {L}angevin {M}onte {C}arlo
              with inaccurate gradient},
   JOURNAL = {Stochastic Process. Appl.},
  FJOURNAL = {Stochastic Processes and their Applications},
    VOLUME = {129},
      YEAR = {2019},
    NUMBER = {12},
     PAGES = {5278--5311},
      ISSN = {0304-4149,1879-209X},
   MRCLASS = {60B10 (60J22 60J70 65C05)},
  MRNUMBER = {4025705},
       DOI = {10.1016/j.spa.2019.02.016},
       URL = {https://doi.org/10.1016/j.spa.2019.02.016},
}

@article{jin2020random,
  title={Random batch methods (RBM) for interacting particle systems},
  author={Jin, Shi and Li, Lei and Liu, Jian-Guo},
  journal={Journal of Computational Physics},
  volume={400},
  pages={108877},
  year={2020},
  publisher={Elsevier}
}

@inproceedings{paulin2025sampling,
  title={{S}ampling from {B}ayesian Neural Network Posteriors with Symmetric Minibatch Splitting {L}angevin Dynamics},
  author={Paulin, Daniel and Whalley, Peter A and Chada, Neil K and Leimkuhler, Benedict J},
  booktitle={International Conference on Artificial Intelligence and Statistics},
  pages={5014--5022},
  year={2025},
  organization={PMLR}
}

@book {ambrosio2008gradient,
    AUTHOR = {Ambrosio, Luigi and Gigli, Nicola and Savar\'e, Giuseppe},
     TITLE = {Gradient flows in metric spaces and in the space of
              probability measures},
    SERIES = {Lectures in Mathematics ETH Z\"urich},
   EDITION = {Second},
 PUBLISHER = {Birkh\"auser Verlag, Basel},
      YEAR = {2008},
     PAGES = {x+334},
      ISBN = {978-3-7643-8721-1},
   MRCLASS = {49-02 (28A33 35K55 35K90 49Q20 60B05)},
  MRNUMBER = {2401600},
MRREVIEWER = {Pietro\ Celada},
}

@book{Roberts2004,
    AUTHOR = {Robert, Christian P. and Casella, George},
     TITLE = {Monte {C}arlo statistical methods},
    SERIES = {Springer Texts in Statistics},
   EDITION = {Second},
 PUBLISHER = {Springer-Verlag, New York},
      YEAR = {2004},
     PAGES = {xxx+645},
}

@book{liu2001monte,
  title={Monte {C}arlo strategies in scientific computing},
  author={Liu, Jun S},
  volume={10},
  year={2001},
  publisher={Springer}
}

@article{roberts1998optimal,
  title={Optimal scaling of discrete approximations to {L}angevin diffusions},
  author={Roberts, Gareth O and Rosenthal, Jeffrey S},
  journal={Journal of the Royal Statistical Society: Series B (Statistical Methodology)},
  volume={60},
  number={1},
  pages={255--268},
  year={1998},
  publisher={Wiley Online Library}
}

@article{schmon2021large,
  title={Large-sample asymptotics of the pseudo-marginal method},
  author={Schmon, Sebastian M and Deligiannidis, George and Doucet, Arnaud and Pitt, Michael K},
  journal={Biometrika},
  volume={108},
  number={1},
  pages={37--51},
  year={2021},
  publisher={Oxford University Press}
}

@article{alenlov2021pseudo,
  title={Pseudo-marginal {H}amiltonian {M}onte {C}arlo},
  author={Alenl{\"o}v, Johan and Doucet, Arnaud and Lindsten, Fredrik},
  journal={Journal of Machine Learning Research},
  volume={22},
  number={141},
  pages={1--45},
  year={2021}
}

@article{gilks1992adaptive,
  title={Adaptive rejection sampling for {G}ibbs sampling},
  author={Gilks, Walter R and Wild, Pascal},
  journal={Journal of the Royal Statistical Society: Series C (Applied Statistics)},
  volume={41},
  number={2},
  pages={337--348},
  year={1992},
  publisher={Wiley Online Library}
}

@book{diggle2002analysis,
  title={Analysis of Longitudinal Data},
  author={Diggle, Peter J. and Heagerty, Patrick J. and Liang, Kung-Yee and Zeger, Scott L.},
  edition={Second},
  year={2002},
  publisher={Oxford University Press}
}

@article{sommer1983increased,
  title={Increased mortality in children with mild vitamin A deficiency},
  author={Sommer, Alfred and Hussaini, Gusti and Tarwotjo, Ignatius and Susanto, Djoko},
  journal={The Lancet},
  volume={322},
  number={8350},
  pages={585--588},
  year={1983},
  publisher={Elsevier}
}

@article {Chen-N-W-22,
    AUTHOR = {Chen, Hong-Bin and Niles-Weed, Jonathan},
     TITLE = {Asymptotics of smoothed {W}asserstein distances},
   JOURNAL = {Potential Anal.},
  FJOURNAL = {Potential Analysis. An International Journal Devoted to the
              Interactions between Potential Theory, Probability Theory,
              Geometry and Functional Analysis},
    VOLUME = {56},
      YEAR = {2022},
    NUMBER = {4},
     PAGES = {571--595},
      ISSN = {0926-2601,1572-929X},
   MRCLASS = {58J65 (28A33 60J60)},
  MRNUMBER = {4396827},
MRREVIEWER = {Dejun\ Luo},
       DOI = {10.1007/s11118-020-09895-9},
       URL = {https://doi.org/10.1007/s11118-020-09895-9},
}

@article{beyler2025convergence,
  title={Convergence of {D}eterministic and {S}tochastic {D}iffusion-{M}odel {S}amplers: {A} {S}imple {A}nalysis in {W}asserstein {D}istance},
  author={Beyler, Eliot and Bach, Francis},
  journal={arXiv preprint arXiv:2508.03210},
  year={2025}
}

@article{eldan2013thin,
  author  = {Eldan, Ronen},
  title   = {Thin Shell Implies Spectral Gap Up to Polylog via a Stochastic Localization Scheme},
  journal = {Geometric and Functional Analysis},
  volume  = {23},
  number  = {2},
  pages   = {532--569},
  year    = {2013}
}

@article{lee2017eldans,
  author  = {Lee, Yin Tat and Vempala, Santosh S.},
  title   = {Eldan's Stochastic Localization and the {KLS} Hyperplane Conjecture: An Improved Lower Bound for Expansion},
  journal = {Proceedings of the IEEE 58th Annual Symposium on Foundations of Computer Science},
  pages   = {998--1007},
  year    = {2017}
}

@article{srinivasan2025poisson,
  title={Poisson Midpoint Method for Log Concave Sampling: Beyond the Strong Error Lower Bounds},
  author={Srinivasan, Rishikesh and Nagaraj, Dheeraj},
  journal={arXiv preprint arXiv:2506.07614},
  year={2025}
}

@article{chen2026high,
  title={High-accuracy log-concave sampling with stochastic queries},
  author={Chen, Fan and Chewi, Sinho and Daskalakis, Constantinos and Rakhlin, Alexander},
  journal={arXiv preprint arXiv:2602.14342},
  year={2026}
}

@article{lu2026sharp,
  title={A sharp hypocoercive entropy decay estimate for underdamped {L}angevin dynamics},
  author={Lu, Jianfeng},
  journal={arXiv preprint arXiv:2605.01933},
  year={2026}
}

@article{louis1982finding,
  title={Finding the observed information matrix when using the {EM} algorithm},
  author={Louis, Thomas A},
  journal={Journal of the Royal Statistical Society Series B: Statistical Methodology},
  volume={44},
  number={2},
  pages={226--233},
  year={1982},
  publisher={Oxford University Press}
}

@article {ChenEldanSL,
    AUTHOR = {Chen, Yuansi and Eldan, Ronen},
     TITLE = {Localization schemes: a framework for proving mixing bounds
              for {M}arkov chains},
   JOURNAL = {Duke Math. J.},
  FJOURNAL = {Duke Mathematical Journal},
    VOLUME = {174},
      YEAR = {2025},
    NUMBER = {8},
     PAGES = {1431--1510},
      ISSN = {0012-7094,1547-7398},
   MRCLASS = {60J22 (60H30 60K35 82M60)},
  MRNUMBER = {4916109},
MRREVIEWER = {Udrea\ P\u aun},
       DOI = {10.1215/00127094-2024-0063},
       URL = {https://doi.org/10.1215/00127094-2024-0063},
}

@inproceedings{klartag2023spectral,
  title={Spectral monotonicity under {G}aussian convolution},
  author={Klartag, Bo’az and Putterman, Eli},
  booktitle={Annales de la Facult{\'e} des sciences de Toulouse: Math{\'e}matiques},
  volume={32},
  number={5},
  pages={939--967},
  year={2023}
}

@article{el2022information,
  title={An information-theoretic view of stochastic localization},
  author={El Alaoui, Ahmed and Montanari, Andrea},
  journal={IEEE Transactions on Information Theory},
  volume={68},
  number={11},
  pages={7423--7426},
  year={2022},
  publisher={IEEE}
}

@article{li2025sharp,
  author  = {Li, Lei and Wang, Yuliang},
  title   = {A sharp uniform-in-time error estimate for stochastic gradient {L}angevin dynamics},
  journal = {CSIAM Transactions on Applied Mathematics},
  volume  = {6},
  number  = {4},
  pages   = {711--759},
  year    = {2025},
  doi     = {10.4208/csiam-am.SO-2024-0039}
}

@article{paulin2026sgld,
  author  = {Paulin, Daniel and Whalley, Peter A.},
  title   = {Theoretical guarantees for stochastic gradient {L}angevin dynamics},
  journal = {arXiv preprint arXiv:2610.01651},
  year    = {2026}
}

@article{guo2005mutual,
  author  = {Guo, Dongning and Shamai, Shlomo and Verd{\'u}, Sergio},
  title   = {Mutual information and minimum mean-square error in {G}aussian channels},
  journal = {IEEE Transactions on Information Theory},
  volume  = {51},
  number  = {4},
  pages   = {1261--1282},
  year    = {2005}
}

\end{document}